\documentclass{LMCS}
\usepackage{amsmath,amssymb,amsfonts,amsfonts,stmaryrd}
\usepackage{enumerate}
\usepackage{bbm}
\usepackage{bussproofs}
\usepackage{graphics}
\usepackage{hyperref}
\usepackage[all]{xy}
\newtheorem{lemmadef}[thm]{Lemma and Definition}{\bfseries}{\upshape}
\newcounter{blubber}
\newenvironment{sparenumerate}
{\begin{list}
  {\arabic{blubber}.}
  {\usecounter{blubber}
		\setlength{\topsep}{0pt}
		\setlength{\partopsep}{0pt}
   \setlength{\leftmargin}{0pt}
    \setlength{\parsep}{0pt}
    \setlength{\itemindent}{4ex}
    \setlength{\itemsep}{2pt}
  }
}
{\end{list}}

\newenvironment{sparitemize}
{\begin{list}{$\bullet$}{
		\setlength{\topsep}{0pt}
		\setlength{\partopsep}{0pt}
    \setlength{\leftmargin}{0pt}
    \setlength{\parsep}{0pt}
    \setlength{\itemindent}{4ex}
    \setlength{\itemsep}{0pt}
  }
}{\end{list}}

\newcommand{\Succ}{\mathsf{Suc}}
\newcommand{\Chld}{\mathsf{Chld}}
\newcommand{\Bp}{\mathcal{B}}

\newcommand{\subf}{\mathsf{subf}}
\newcommand{\eat}[1]{}

\newcommand{\Ord}{\mathcal{O}}

\renewcommand{\mod}{\mathsf{m}}
\newcommand{\fix}{\mathsf{f}}

\newcommand{\Tr}{\mathsf{Tr}}
\newcommand{\Z}{\mathbb{Z}}

\newcommand{\cf}{\mathbbm{1}}

\newcommand{\pto}{\rightharpoonup}

\newcommand{\Mon}{\mathcal{M}}
\newcommand{\Rat}{\mathbb{Q}}
\newcommand{\Nat}{\mathbb{N}}
\newcommand{\supp}{\mathsf{supp}}
\newcommand{\inv}{^{-1}}
\newcommand{\size}{\mathsf{size}}

\newcommand{\LFP}{\mathsf{LFP}}
\newcommand{\GFP}{\mathsf{GFP}}

\renewcommand{\Game}{\mathsf{G}}
\newcommand{\Dist}{\mathsf{D}}
\newcommand{\Bag}{\mathsf{B}}

\newcommand{\Rules}{\mathsf{R}}
\newcommand{\FoRm}{\mathcal{F}}

\newcommand{\Pow}{\mathcal{P}}
\newcommand{\Seq}{\mathsf{S}}

\newcommand{\B}[1]{[\lambda]}
\newcommand{\game}{\mathcal{G}}
\newcommand{\atomic}{\mathsf{At}}

\newcommand{\F}{T}
\newcommand{\lsem}{[\![}
\newcommand{\rsem}{]\!]}

\newcommand{\Var}{\mathsf{V}}
\newcommand{\hearts}{\heartsuit}
\renewcommand{\bar}{\overline}

\newcommand{\Set}{\mathsf{Set}}
\newcommand{\op}{^\mathrm{op}}

\newcommand{\Tab}{\mathsf{T}}

\newcommand{\psem}[1]{\lsem#1\rsem}

\newcommand{\bbY}{\mathbb{Y}}
\newcommand{\length}{\mathrm{length}}
\newcommand{\nil}{\epsilon}
\newcommand{\Rel}{\mathrm{Rel}}
\newcommand{\bbA}{\mathbb{A}}
\newcommand{\Cl}{\mathrm{Cl}}

\newcommand{\model}{\mathit{M}}
\newcommand{\MC}{\mathcal{MG}}

\newcommand{\Diag}{\mathsf{Diag}}

\newcommand{\pstrat}{g}

\def\doi{7 (3:03) 2011}
\lmcsheading%
{\doi}
{1--33}
{}
{}
{Apr.~26, 2010}
{Aug.~10, 2011}
{}

\begin{document}

\title[\textsc{Exptime} Tableaux for the Coalgebraic
$\mu$-calculus]{\textsc{Exptime} Tableaux for the Coalgebraic
$\mu$-Calculus\rsuper*}

\author[C.~C\^irstea]{Corina C\^irstea\rsuper a}	%required
\address{{\lsuper a}School of Electronics and Computer Science, University of
Southampton}	%required
\email{cc2@ecs.soton.ac.uk}  %optional
%\thanks{thanks 1, optional.}	%optional

\author[C.~Kupke]{Clemens Kupke\rsuper b}	%required
\address{{\lsuper{b,c}}Department of Computing, Imperial College London}	%required
\email{ckupke@doc.ic.ac.uk, dirk@doc.ic.ac.uk}  %optional
\thanks{{\lsuper b}Partially supported by grant EP/F031173/1 from the UK EPSRC}	%optional

\author[D.~Pattinson]{Dirk Pattinson\rsuper c}	%required
\address{\vskip-6 pt}	%required
%\email{dirk@doc.ic.ac.uk}  %optional
%\thanks{}	%optional

\titlecomment{{\lsuper*} This paper is a revised and extended version of the paper ``EXPTIME Tableaux for Coalgebraic $\mu$-calculi'' that has been published in the proceedings of CSL 2009 ~\cite{cikupa:expt09}.}

\begin{abstract}
  The coalgebraic approach to modal logic provides a uniform framework
  that captures the semantics of a large class of structurally
  different modal logics, including e.g. graded and probabilistic
  modal logics and coalition logic.  In this paper, we introduce the
  coalgebraic $\mu$-calculus, an extension of the general
  (coalgebraic) framework with fixpoint operators. Our main results
  are completeness of the associated tableau calculus and
	\textsc{Exptime}
  decidability  for guarded formulas.  Technically, this is achieved by reducing
  satisfiability to the existence of non-wellfounded tableaux, which
  is in turn equivalent to the existence of winning strategies in
  parity games.
  Our results are parametric in the underlying class of models and
  yield, as concrete applications, previously unknown complexity
  bounds for the probabilistic $\mu$-calculus and for an extension of
  coalition logic with fixpoints.
\end{abstract}

\maketitle

%------------------------------------------------------------------------- 

\section{Introduction}

The extension of a modal logic with operators for least and greatest
fixpoints leads to a dramatic increase in expressive power
\cite{Bradfield:1996:EMM}. The
paradigmatic example is of course the modal $\mu$-calculus
\cite{Kozen:1983:RPC}. In the same way that the $\mu$-calculus
extends the modal logic $K$, one can freely add fixpoint operators to any
propositional modal logic, as long as modal operators are monotone.
Semantically, this poses no problems, and the interpretation of
fixpoint formulas can be  defined in a standard way in terms of the
semantics of the 
underlying modal logic.

This apparent simplicity is lost once we move from semantics to
syntax: completeness and complexity even of the modal $\mu$-calculus
are all but trivial 
\cite{Walukiewicz:2000:CKA,Emerson:1988:CTA}, and $\mu$-calculi arising
from other monotone modal logics are largely unstudied, with the
notable exception of the graded $\mu$-calculus
\cite{Kupferman:2002:CGM}. Here, we improve on this situation, not
by providing a new complexity result for a specific fixpoint logic,
but by providing a generic and uniform treatment of modal fixpoint logics on the
basis of \emph{coalgebraic semantics}.
This allows for a generic and
uniform treatment of a large class of modal logics and replaces
the investigation of a concretely given
logic with the study of \emph{coherence
conditions} that  mediate between the axiomatisation and the
(coalgebraic) semantics. The
use of coalgebras conveniently abstracts the details of a concretely
given class of models, which is replaced by the class of coalgebras
for an (unspecified) endofunctor on sets. Specific choices for this
endofunctor then yield specific model classes, such as the class of
all Kripke frames or probabilistic transition systems.
A property such as completeness or complexity of a specific logic is
then automatic once the coherence conditions are satisfied. As it
turns out, even \emph{the same} coherence conditions that guarantee
completeness and decidability of the underlying modal logic entail
the same properties of the ensuing $\mu$-calculus. This immediately
provides us with a number of concrete examples: as instances of the
generic framework, we obtain not only the known \textsc{Exptime} bounds, both for the
modal and the graded $\mu$-calculus
\cite{Emerson:1988:CTA,Kupferman:2002:CGM}, but also previously
unknown \textsc{Exptime} bounds for the probabilistic and monotone
$\mu$-calculus, and for an extension of coalition logic
\cite{Pauly:2002:MLC} with fixpoint operators.

Our main technical results are
a syntactical characterisation of satisfiability in terms of
(non-) ex\-istence of closed tableaux and a game-theoretic
characterisation of satisfiability that yields an \textsc{Exptime} upper bound for
the satisfiability problem for guarded formulas. Along the way, we establish a small model theorem. Here, as usual, a formula
is called guarded if every fixpoint variable occurs
only within the scope of a modal operator.  If we assume that every formula can be transformed into an equivalent guarded formula in polynomial time, our EXPTIME decidability result extends to the full coalgebraic $\mu$-calculus. This assumption is generally made in the literature on the modal $\mu$-calculus  \cite{Kupferman:2000:ATA}, but a recent paper \cite{frla:them11} argues that in fact no algorithm is known 
that can perform the transformation in polynomial time. Therefore we formulate our EXPTIME-decidability result more restrictive than in~\cite{cikupa:expt09}. We nevertheless conjecture that our tableau calculus can be used for proving EXPTIME-decidability for the full coalgebraic $\mu$-calculus.
%as we will explain
%in the conclusions.

We start by describing a parity game that characterizes model checking for the coalgebraic $\mu$-calculus. As in the model-checking game for the modal $\mu$-calculus (see e.g.~\cite{Stirling:2001:MTP}), we allow greatest and least fixpoints to be unfolded ad libitum. Truth of a formula in a particular state of a model then follows, if only greatest fixpoints
are unfolded infinitely often on the top level along infinite paths,
which is captured by a parity condition.
The same technique is employed in the
construction of tableaux, which we conceptualise as finite directed
graphs: closed tableaux witness unsatisfiability of the root
formula, provided that along any infinite tableau path one can construct
an infinite sequence of formulas (a trace that tracks the evolution
of formulas in a tableau) that violates the parity condition.
In particular, 
closed tableaux are finitely represented proofs of the
unsatisfiability of the root formula.
Soundness of the tableau
calculus is established by showing that a winning strategy  in
the model checking game precludes existence of a closed tableau.
Decidability is
then established with the help of tableau games, where the adversary
chooses a tableau rule, and the player claiming satisfiability
chooses one conclusion which effectively constructs a path in a
tableau. In order to turn this tableau game into a parity game we combine
the game board with the transition function of a deterministic parity word automaton.
This automaton checks that on any given play, i.e., on any tableau path,
there exists no trace that violates the parity condition. 
%greatest fixpoints are unfolded infinitely often. 
We prove adequacy
of the tableau game by constructing a satisfying model from a
winning strategy in the tableau game, which makes crucial use of the
coherence conditions between the axiomatisation and the coalgebraic
semantics.
This allows us to determine satisfiability of a fixpoint formula by
deciding the associated (parity) tableau game, and the announced
\textsc{Exptime} upper bound for guarded formulas 
follows once we can ensure that legality of
moves in the tableau game can be decided in exponential time.

\textbf{Related Work.}
Our treatment is
inspired by
\cite{Niwinski:1996:GMC,Venema:2006:AFP,Schroder:2008:PBR}, but
we note some important differences. In contrast to
\cite{Niwinski:1996:GMC}, we use parity games that directly
correspond to tableaux, together with  parity automata to detect bad
traces. Moreover, owing to the generality of the coalgebraic
framework, the model construction here needs to super-impose a coalgebra
structure on the relation induced by a winning strategy.
This construction is necessarily different from 
\cite{Schroder:2008:PBR}, since we cannot argue
by induction on modal rank in the presence of fixpoints. 
Coalgebraic fixpoint logics are also treated  in
\cite{Venema:2006:AFP}, where an automata theoretic characterisation
of satisfiability is presented. We add to this picture by
providing complexity results and a complete tableau calculus.
Moreover, we use standard syntax for modal operators, which allows us to
subsume for instance the graded $\mu$-calculus that cannot be
expressed in terms of the $\nabla$-operator used in \emph{op.cit.}

\section{The Coalgebraic $\mu$-Calculus}

\noindent
To keep our treatment fully parametric in the underlying (modal)
logic,  we define the syntax of the coalgebraic $\mu$-calculus
relative to a (fixed) modal 
similarity type, that is, a set $\Lambda$
of modal operators with associated arities. Throughout, we fix a denumerable set
$\Var$ of propositional variables.  We will only deal with formulas
in 
negation normal form and abbreviate
$\bar\Lambda = \lbrace \bar \hearts \mid \hearts \in \Lambda \rbrace$ and
$\bar\Var = \lbrace \bar p \mid p \in \Var \rbrace$. The
arity of $\bar \hearts \in \bar \Lambda$ is the same as that of
$\hearts$.
The set $\FoRm(\Lambda)$ of $\Lambda$-formulas is given by the
grammar 
\[
  A, B ::= p  \mid A \lor B \mid A \land B \mid
	\hearts(A_1, \dots, A_n) \mid \mu p. A \mid \nu p. A \]
where $p \in \Var \cup \bar \Var$, $\hearts \in \Lambda \cup \bar
\Lambda$ is $n$-ary  and $\bar p$ does not occur in $A$ in
the last two clauses. The sets of free and bound variables of a
formula are defined as usual, in particular $p$ is bound
in $\mu p. A$ and $\nu p. A$. Negation $\bar{\,\cdot\,}:
\FoRm(\Lambda) \to \FoRm(\Lambda)$ is given inductively by
$\bar{\bar p} = p$, $\bar{A \land B} = \bar A
\lor \Bar B$, $\bar {\hearts(A_1, \dots, A_n)} = \bar
\hearts(\bar{A_1}, \dots, \bar{A_n})$ and $\bar{\mu p. A} = \nu p. \bar A
[ \bar p := p]$ and the dual clauses for $\lor$ and $\nu$.
If $S$ is a set of formulas, then the
collection of formulas
that arises by prefixing elements of $S$ by one layer
of modalities is denoted by
$(\Lambda \cup \bar \Lambda)(S) = \lbrace \hearts(A_1,\dots, A_n)
\mid \hearts \in \Lambda \cup \bar\Lambda \;  \mbox{$n$-ary,}~ A_1, \dots, A_n \in S \rbrace$. 
A \emph{substitution} is a mapping $\sigma: V
\to \FoRm(\Lambda)$ and $A \sigma$ is the result of replacing all
free occurrences of $p \in V$ in $A$ by  $\sigma(p)$.

On the semantical side, parametricity is achieved by adopting
coalgebraic semantics: formulas are interpreted over $T$-coalgebras,
where $T$ is an (unspecified) endofunctor on sets, and we recover
the semantics of a large number of logics in the form of specific
choices for $T$. To interpret the modal operators $\hearts \in
\Lambda$, we require that $T$ extends to a 
\emph{$\Lambda$-structure} and comes with a
predicate lifting, that is, a natural
transformation of type
$\lsem \hearts \rsem: 2^n \to 2 \circ T\op$
for every $n$-ary modality $\hearts \in \Lambda$, where $2: \Set \to
\Set\op$ is the contravariant powerset functor. In elementary terms,
this amounts to assigning a
set-indexed family of functions $(\lsem \hearts \rsem_X: \Pow(X)^n \to
\Pow(TX))_{X \in \Set}$  to every $n$-ary modal operator $\hearts
\in \Lambda$
such that
$(Tf)\inv \circ \lsem \hearts \rsem_X(A_1, \dots, A_n) = \lsem
\hearts \rsem_Y (f\inv(A_1), \dots, f\inv(A_n))$ for all functions $f: Y \to
X$. If $\hearts \in \Lambda$ is $n$-ary, we put
$\lsem
\bar\hearts\rsem_X(A_1, \dots, A_n) = (TX) \setminus \lsem \hearts
\rsem_X (X \setminus A_1, \dots,  X \setminus A_n)$.
We usually denote a
structure by the endofunctor $T$ and leave the definition of
the predicate liftings implicit.  A $\Lambda$-structure is
\emph{monotone} if, for all sets $X$ we have that
$\lsem \hearts \rsem_X (A_1, \dots, A_n) \subseteq \lsem
\hearts \rsem_X(B_1, \dots, B_n)$ whenever $A_i \subseteq B_i$ for all $i
= 1, \dots, n$.

In the coalgebraic approach, the role of frames is played by
\emph{$T$-coalgebras}, 
i.e.~pairs $(X, \gamma)$ where $X$ is a set (of states) 
and $\gamma: X \to TX$ is a (transition) function. A
\emph{$T$-model} is a triple $(X, \gamma, h)$ where $(X,
\gamma)$ is a $T$-coalgebra and $h: \Var \to \Pow(X)$ is a
valuation 
of the propositional variables that we implicitly extend to $\Var
\cup \bar \Var$ by putting
$h(\bar p) = X \setminus h(p)$.
For a monotone $\Lambda$-structure $T$ and a $T$-model
$M = (X, \gamma, h)$, the \emph{truth set} $\lsem A \rsem_M$ of a formula $A \in
\FoRm(\Lambda)$ w.r.t. $M$ is given inductively by
\begin{gather*}
\lsem p \rsem_M  = h(p)  \quad
%\lsem \bar p \rsem_M   = C \setminus \pi(p) \quad
\lsem \mu p. A \rsem_M  = \LFP (A^M_p)  \quad
\lsem \nu p. A \rsem_M  = \GFP (A^M_p)   \\[1ex]
\lsem \hearts(A_1, \dots, A_n) \rsem_M = \gamma\inv \circ \lsem
\hearts \rsem_X (\lsem A_1 \rsem_M, \dots, \lsem A_n \rsem_M) 
\end{gather*}
where $\LFP (A^M_p)$ and $\GFP (A^M_p)$ are the least and greatest fixpoint of the 
monotone mapping $A^M_p : \Pow(X) \to \Pow(X)$ defined by $A^M_p(U) =
\lsem A \rsem_{M'}$ where $M' = (X, \gamma, h')$ and $h'(q) = h(q)$
for $q \neq p$ and $h'(p) = U$. We write
$M, x \models A$ if $x \in \lsem A \rsem_M$ to denote
that $A$ is satisfied at $x$.
A formula $A \in
\FoRm(\Lambda)$ is \emph{satisfiable} w.r.t.~a given
$\Lambda$-structure $T$ if there exists a $T$-model $M$ such
that $\lsem A \rsem_M \neq \emptyset$. 
The mappings $A^M_p$ are indeed monotone in
case of a monotone $\Lambda$-structure, which guarantees the
existence of fixpoints.

\begin{exa}
\label{example:logics}
\begin{sparenumerate}
\item $T$-coalgebras $(X, \gamma: X \to \Pow(X))$ for $TX = \Pow(X)$
are Kripke frames. If 
$\Lambda = \lbrace \Box \rbrace$ for $\Box$ unary and $\bar \Box =
\Diamond$,
$\FoRm(\Lambda)$ are the formulas of the modal $\mu$-calculus
\cite{Kozen:1983:RPC}, and the structure
$\lsem \Box \rsem_X(U) = \lbrace V \in \Pow(X) \mid
V \subseteq U \rbrace$ gives its semantics.
\item  The syntax of the graded $\mu$-calculus
\cite{Kupferman:2002:CGM} is given (modulo an index shift)
by the similarity type
$\Lambda = \lbrace \langle n \rangle \mid n \geq 0 \rbrace$ where
$\bar{\langle n \rangle} = [n]$, and $\langle n
\rangle A$ reads as ``$A$ holds in more than $n$ successors''.
In contrast to \emph{op.~cit.} we
interpret the graded $\mu$-calculus over multigraphs,
i.e.~coalgebras for the functor $\Bag$ 
\[ \Bag (X) = \lbrace f: X \to \Nat \mid \supp(f) \mbox{ finite}
\rbrace
\]
where $\supp(f) = \lbrace x \in X \mid f(x) \neq 0 \rbrace$ is the
support of $f$,
that extends to a structure 
\[ \lsem \langle n \rangle \rsem_X(U) = \lbrace f \in \Bag(X) \mid
\sum_{x \in U} f(x) > n \rbrace \qquad \mbox{for} \; U \subseteq X. \]
Note that this semantics differs from the Kripke semantics for both
graded modal logic \cite{Fine:1972:MPW} and the graded $\mu$-calculus.
The change of the semantics is needed in order to fit graded modal logic
into the coalgebraic framework, because in the standard semantics of 
graded modal logic we cannot interpret the modalities by natural transformations.
Both types of semantics, however, induce the same satisfiability problem:
image-finite Kripke frames are multigraphs where each edge has multiplicity one,
and the unravelling of a multigraph can be turned into a Kripke
frame by inserting the appropriate number of copies of each state.
The transformations preserve satisfiability. 
The fact that the two types of semantics induce the same satisfiability problem makes use of the fact 
that the graded $\mu$-calculus has the tree-model property (\cite{Kupferman:2002:CGM}): 
a formula of the graded $\mu$-calculus is satisfiable on some Kripke frame iff it is satisfiable on a tree of finitely bounded branching
degree. Alternatively, the fact that the two satisfiability problems are equivalent can be also obtained from the results in this paper by showing that the tableau calculus for the graded $\mu$-calculus is \emph{sound} over the class of all Kripke frames.
\item The probabilistic $\mu$-calculus arises from
the similarity 
type $\Lambda = \lbrace \langle p \rangle \mid p \in [0, 1] \cap \Rat \rbrace$
where $\bar{\langle p \rangle} = [p]$ and $\langle p \rangle \phi$ reads as ``$\phi$ holds with
probability at least $p$ in the next state''. 
The semantics of the
probabilistic $\mu$-calculus is given by the structure
\[ \Dist (X) = \lbrace \mu: X \to_{\! f} [0, 1] \mid 
         \sum_{x \in X} \mu (x) = 1 \rbrace \quad
   \lsem \langle p \rangle  \rsem_X (U) = \lbrace \mu  \in  \Dist(X) \mid \sum_{x
\in U} \mu(x) \geq p \rbrace \]
where $U \subseteq X$ and $\to_{\! f}$ indicates maps with finite support.
Coalgebras for $\Dist$ are precisely image-finite Markov chains, and
the finite model property of the coalgebraic $\mu$-calculus that we
establish later
ensures that satisfiability is independent of 
image-finite semantics.
\item Formulas of coalition logic over a finite set $N$ of agents
\cite{Pauly:2002:MLC} arise via $\Lambda = \lbrace
[C ] \mid C \subseteq N \rbrace$, and are interpreted over game frames, i.e.~coalgebras for the functor 
\[ \Game (X) = \lbrace (f, (S_i)_{i \in N}) \mid \prod_{i \in N} S_i
\neq \emptyset, f: \prod_{i \in N} S_i \to X \rbrace
\]
which is a class-valued functor, which
however fits with the subsequent development. We think of $S_i$ as
the
set of strategies for agent $i$ and $f$ is an outcome function. 
The formula $[C] A$ reads as ``coalition $C$ can achieve $A$'', which is captured
by the lifting
\[
  \lsem [C] \rsem_X(U) = \lbrace (f, (S_i)_{i \in N}) \in \Game (X)
	\mid
	\exists (s_i)_{i \in C} \forall (s_i)_{i \in N \setminus C} 
	f((s_i)_{i \in N}) \in U \rbrace
\]
for $U \subseteq X$. The induced coalgebraic semantics is precisely 
the standard semantics of coalition logic, ie., the formula $[C] A$ holds at a state $x$
if all agents $i$ in the coalition $C$ can choose a strategy $s_i$ at $x$ such that,
for all possible strategy choices of agents in $N \setminus C$ at position $x$, 
the play proceeds to a state $x'$ that satisfies property $A$.
\item Finally, the similarity type $\Lambda = \lbrace \Box \rbrace$ 
of monotone modal logic
\cite{Chellas:1980:ML} has a single unary $\Box$ (we write $\bar
\Box = \Diamond$) and interpret the ensuing language  
over monotone neighbourhood frames, that is,
coalgebras for the functor  / structure
\[ \Mon (X) = \lbrace Y \subseteq \Pow(\Pow(X)) \mid Y \mbox{ upwards
closed} \rbrace \quad
   \lsem \Box \rsem_X(U) = \lbrace Y \in \Mon(X) \mid U \in Y
\rbrace \]
for $U \subseteq X$ which recovers the standard semantics in a coalgebraic setting
\cite{Hansen:2004:CPM}.
\end{sparenumerate}
\noindent
It is readily verified that all structures above
are indeed monotone. 
\end{exa}

\section{The Model-Checking Game}\label{sec:modcheck}

We start by characterising the satisfaction relation between states
of a model and formulas of the coalgebraic $\mu$-calculus in terms
of a two-player parity game that we call the \emph{model checking
game}. This characterisation will be the main
technical tool for establishing soundness and completeness of an
associated tableau calculus.

The game that we are about to describe generalises
\cite[Theorem~1, Chapter~6]{Stirling:2001:MTP} to the coalgebraic
setting, and is a variant of the game used in 
~\cite{Cirstea:2008:CMCS}. We begin by fixing our terminology
concerning parity games.

A \emph{parity game} played by $\exists$
(\'Eloise) and $\forall$ (Ab\'elard) is a tuple $\game = (B_\exists, B_\forall, E,
\Omega)$ where 
%\begin{sparitemize}
%\item 
$B = B_\exists \cup B_\forall$ is the disjoint union of
\emph{positions} owned by $\exists$ and $\forall$, respectively, 
$E \subseteq B \times B$ %is a binary relation on $B$ that
indicates the allowed moves, and
$\Omega: B \to \omega$ is a (parity) map with finite range.
An 
%finite or 
infinite sequence $(b_0, b_1, b_2, \dots)$ of
positions 
is called \emph{bad} if $\max \lbrace k \mid k = \Omega(b_i) 
\mbox{ for infinitely many $i \in \omega$} \rbrace$ is odd. 

A \emph{play} in $\game$
is a finite or infinite sequence of positions $(b_0, b_1,
\dots)$ with the property that $(b_{i}, b_{i+1}) \in E$ for all $i$, i.e.~all
moves are legal, and $b_0$ is the \emph{initial
position} of the play. 
A \emph{full play} is either 
infinite, or a finite play ending in a position $b_n$ where 
$E[b_n] = \lbrace b \in B \mid (b_n, b) \in E \rbrace  = \emptyset$, i.e.~no 
more moves are possible. A finite play is
lost by the player who cannot move, and an infinite play $(b_0,
b_1, \dots)$ is lost by $\exists$ (and won by $\forall$) 
iff $(b_0, b_1, \dots)$ is bad.

A {\em strategy} in $\game$ for a player $P \in \{ \exists ,
\forall\}$ is a partial function 
that maps all plays that end in a position $b \in B_P$ of $P$ 
with $E[b] \neq \emptyset$ to
a position $b' \in B$ such that $(b, b') \in E$.
Intuitively, a strategy determines a player's next move, depending
on the history of the game in all positions where the player can
move. Given a strategy $s$ for player $P$ in $\game$ we say that
a play $(b_0,\dots, b_i, \dots)$ of $\game$  is \emph{played according to} 
$s$ if for all proper prefixes $b_0 \ldots b_i$ of $\pi$ with
$b_i \in B_P$ we have $s(b_0 \ldots b_i) = b_{i+1}$.
A strategy for a player $P \in \{\exists,\forall\}$ is called
{\em history-free} or \emph{positional} if it only depends on the last position
of a play.
Formally, a \emph{history-free strategy} for player $P \in \lbrace \exists,
\forall \rbrace$ is a partial function $s: B_P \pto B$
such that $s(b)$ is defined iff $E[b] \neq \emptyset$, in which case $(b, s(b))
\in E$. A play $(b_0, b_1, \dots)$ is
\emph{played according to $s$} if $b_{i+1} = s(b_i)$
for all $i$ with $b_i \in B_P$, and $s$ is a \emph{winning strategy}
from position $b \in B$ if $P$ wins all plays with initial position
$b$ that are played according to $s$.

It is known that parity games 
are history-free determined \cite{emer:tree91,most:game91} 
and that winning
regions can be decided in 
 $\mbox{UP} \cap \mbox{co-UP}$ \cite{jurd:small00}.
 
\begin{thm}{\cite{jurd:small00}}\label{fact:paritygames}
 At every position $b \in B_\exists \cup B_\forall$ in a parity game 
 $\game = (B_\exists,B_\forall,E,\Omega)$
 one of the players has a history-free winning strategy. Furthermore,
 for every $b \in B_\exists \cup B_\forall$, 
 it can be determined in
	time $O\left(d \cdot m \cdot \left( \frac{n}{\lfloor d/2 \rfloor}
	\right)^{\lfloor d/2 \rfloor}\right)$  which player
	has a winning strategy from position $b$, where
$n$, $m$ and $d$ are the size of $B$, $E$ and the range of
$\Omega$, respectively.
\end{thm}
%\noindent
We will now introduce the model checking game as a parity game.
The model checking game is played on pairs $(A, x)$ where $A$ is a formula
and $x$ is a state, and (informally) $\forall$ tries to demonstrate
that $x \not\models A$ whereas $\exists$ claims the opposite. The
formulation of the game
relies on formulas being \emph{clean} (no variable occurs both free and bound,
 or is bound more than once)
and
\emph{guarded} (bound variables only occur within the scope of modal
operators). In the model checking game, we will only encounter a
finite set of formulas, those that lie in the \emph{closure} of the initial
formula. 
%We will argue later
%that every formula can be transformed to an equivalent clean and
%guarded formula in linear time. 
The size of the closure will play a 
crucial role in our main complexity result because it yields an upper bound for 
the size of our tableau game that characterizes {\em satisfiability} of a formula. The formal
definitions are as follows:
\begin{defi}
A set $\Gamma \subseteq \FoRm(\Lambda)$ of formulas is \emph{closed}
if $B \in \Gamma$ whenever $B$ is a subformula of some $A \in
\Gamma$ and 
$A[p := \eta p. A] \in \Gamma$ if $\eta p. A \in \Gamma$,
where $\eta \in \lbrace
\mu, \nu \rbrace$. The \emph{closure} of $\Gamma$
is the smallest closed set $\Cl(\Gamma)$ for which 
$\Gamma \subseteq \Cl(\Gamma)$.
 
A formula $A \in \FoRm(\Lambda)$ is \emph{guarded} if, for all
subformulas $\eta p. B$ of $A$, $p$ only appears in the scope of a
modal operator within $B$,
and $A$ is \emph{clean} if the sets of free and bound variables of a
formula are  disjoint and if no two distinct occurrences of fixpoint operators
in $A$ bind the same variable.
A finite set of formulas $\Gamma$ is guarded if every element of $\Gamma$
is guarded and $\Gamma$ is clean if the formula $\bigwedge_{A \in \Gamma}
A$ is clean.
\end{defi}
\noindent
In the model checking game, the unfolding of fixpoint formulas gives
rise to infinite plays, and we have to ensure that 
all infinite plays that cycle on an
outermost $\mu$-variable are lost by $\exists$ (who claims
that the formula(s) under consideration are satisfied), as this would
correspond to the infinite unfolding of a least fixpoint. This is
where the parity map comes in: formulas of the form $\mu p. A$ are
assigned odd priorities and, dually, $\nu A.p$ an even priority. To
make sure that $\exists$ only looses those plays that cycle on the
unfolding of an \emph{outermost} $\mu$-variable, we require that the
assignment of priorities is compatible with the 
subformula ordering.
\begin{defi} \label{defn:paritymap}
A \emph{parity map} for a finite, clean set of formulas $\Gamma$ is
a function $\Omega: \Cl(\Gamma) \to \omega$ with finite range for
which
%\begin{sparitemize}
%\item 
$\Omega(A) = 0$ unless $A$ is of the form $\eta p. B$, $\eta
\in \lbrace \mu, \nu \rbrace$, 
$\Omega(A)$ is odd (even) iff $A$ is of the form $\mu p.
B$ ($\nu p. B$), and 
$\Omega(\eta_1 p_1 .B_1) \le \Omega(\eta_2 p_2 .B_2)$ 
whenever
$\eta_1 p_1.B_1$ is a subformula of $\eta_2 p_2.B_2$, where
$\eta_1, \eta_2 \in \lbrace \mu, \nu  \rbrace$.
\end{defi}
\noindent
It is easy to see that every clean set of formulas admits a parity
function.
\begin{lem}
If 
$\Gamma \subseteq \FoRm(\Lambda)$ is finite and clean, then $\Gamma$
admits a parity function whose range is bounded by the cardinality of
$\Cl(\Gamma)$.
\end{lem}
\begin{proof}
By induction on the well-founded ordering generated by
\[ \Gamma, \Delta < \Gamma, A \mbox{ iff } A \notin \Delta
\subseteq \subf(A) \]
where $\subf(A)$ denotes the subformulas of $A$. If $\Gamma$
contains a top-level conjunction, disjunction or propositional
variable, then the claim follows by induction hypothesis. Now
suppose that $\Gamma = \mu p. A, \Gamma'$. By induction hypothesis,
we obtain a parity function $\Omega': \Cl(A, \Gamma') \to \omega$
that we may extend to a parity function $\Omega:
\Cl(\Gamma) \to \omega$ by putting
\[\Omega(B) = \begin{cases}
   m & B = \mu p. A \\
	 \Omega'(B) & B \in \Cl(A, \Gamma') \\
	0 & \mbox{otherwise}
\end{cases}\]	
where $m$ is odd and $m > \Omega'(B)$ for all $B \in \Cl(A, \Gamma')$.
The case $\Gamma = \nu p. A, \Gamma'$ can be treated in a similar fashion. 
\end{proof}
\noindent
Given a parity function, we can define the following parity game,
the winning regions of which characterise satisfiability. We
parametrise the model checking game in a \emph{set} of formulas
which will enable us to use it to prove soundness and completeness
of the tableau calculus (which operates on sets of formulas)
that we introduce later.
\begin{defi}
Suppose that $\model = (X, \gamma, h)$ is a 
$T$-model, $\Gamma \subseteq \FoRm(\Lambda)$ is finite, clean and
guarded, and 
$\Omega$ is a parity map for $\Gamma$. The \emph{model checking
game} $\MC_\Gamma(\model)$ is the parity game whose positions and
admissible moves are given in the following table,
\begin{center} \begin{tabular}{|l|c|l|} \hline 
Position: b & Player & Admissible moves: $E[b]$  \\ \hline 
$(p,x), x \in h(p)$ & $\forall$ & $\emptyset$ \\ 
$(p,x), x \not\in h(p)$ & $\exists$ & $\emptyset$ \\ 
$(\eta p. A(p),x)$ for $\eta \in \lbrace \mu, \nu \rbrace$  &
$\exists$ & $ \{( A[p = \eta p.A(p)], x)\}$ \\ 
%$(\nu p. A(p),c)$ & $\exists$ & $\{A[p =
%\nu p.A(p)]\}$ \\ 
$(A_1 \vee A_2,x)$ & $\exists$ &  $\{( A_1,x),
(A_2,x)\}$ \\ $(A_1 \wedge A_2,x)$ & $\forall$ &  $\{( A_1,x),
(A_2,x)\}$ \\
$(\hearts(A_1,\ldots,A_n),x)$ & $\exists$ & $\{
		(\hearts(A_1,\ldots,A_n),(U_1,\ldots,U_n)) \mid$ \\ & &
		\quad $U_1, \dots, U_n
		    \subseteq X, \gamma(x) \in \lsem
		\hearts \rsem_{X} (U_1,\ldots,U_n)  \rbrace$ \\
$(\hearts (A_1,\ldots,A_n),(U_1,\ldots,U_n))$ & $\forall$
		& $\{(A_i,x) \mid 1 \leq i \leq n, x \in U_i\} $\\ 
%    \cup \{ (\bar\varphi,y') \mid y' \notin U \}$ \\ 
		\hline \end{tabular}  \end{center} 
where $p \in \Var 
\cup \bar \Var$, $\hearts \in \Lambda
\cup \bar\Lambda$,
		$A,A_1,\ldots,A_n \in \Cl(\Gamma)$ are $
		\Lambda$-formulas, $x \in X$ are states and 
   $U_i \subseteq X$ are state sets.
The parity function of $\MC_\Gamma(\model)$ is given by 
%\begin{enumerate}[(i)] \item
$\Omega'(A,x) = \Omega(A)$  for $A \in \Cl(\Gamma)$ and $x \in
		X$, and $\Omega'(\_) = 0$ otherwise. 
%		\end{enumerate} 
	%	\end{itemize} 
\end{defi} 

\noindent
It is easy to see that any two 
parity functions for a given set of
formulas induce the same winning region for both players. We
therefore speak of 
\emph{the} model checking game given by a set of formulas.
Evidently, the model checking game is an extension of the boolean
satisfiability game with fixpoints and modal operators. When the
game reaches a fixpoint formula, that is, a position of type $(\eta
p.A, x)$, this fixpoint is simply unfolded, and its nature (least or
greatest fixpoint) and
nesting depth of the formula fixpoint formula is recorded by the
parity function. To show that a state $x$ satisfies a modal formula
$\hearts(A_1, \dots, A_n)$, $\exists$ needs to select sets $U_1,
\dots, U_n$ (that we think of as a subset of the truth sets of the $A_i$'s) so
that the state $x$ is being mapped  by $\gamma$ into the lifting of $U_1, \dots,
U_n$. Subsequently, $\forall$ may challenge this choice and select
an index $1 \leq i \leq n$ and require that $\exists$ demonstrates
that the formula $A_i$ is satisfied at an aribitrary element of
$U_i$ (and thus corroborate that we may think of $U_i$ as the truth set
of $A_i$). To prove that the model checking game characterises
satisfiability, we make crucial use of monotonicity, as the $U_i$
under-approximate the truth sets of the $A_i$.
The announced generalisation of
\cite[Theorem~1, Chapter~6]{Stirling:2001:MTP} now takes the following
form:
%By induction on the structure of formulas one then proves:
\begin{thm}\label{thm:adequacy}
   For $\Gamma$ finite, clean and guarded, a $T$-model $M = (X,
	 \gamma, h)$, 
	 $A \in \Cl(\Gamma)$ and $x \in X$, 
   $\exists$ has a winning 
   strategy in $\MC_\Gamma(\model)$ from position $(A,x)$
   iff $\model,x \models A$.
\end{thm}
\begin{proof}
The proof is by induction on $A$, and similar to the proof of adequacy of
the game semantics for the coalgebraic $\mu$-calculus
\cite[Theorem~1]{Venema:2006:AFP}. It should be noted that the
model-checking game in \emph{loc. cit.} has slightly diferent moves in positions
that correspond to fixpoint formulas: in a position of the form $(\eta
p.A(p),x)$, the only available choice is to move to $(A(p),x)$, and if a
position of the form $(p,y)$ is reached later, then the only option is to
move to $(A(p),y)$. However, one can show that both ways of treating
fixpoint formulas in the model-checking game are equivalent. We only treat
the case $A =
\hearts(A_1, \dots, A_n)$; all others are as in \emph{loc.~cit.}\,. 

First suppose that $M, x \models \hearts(A_1, \dots, A_n)$. By
induction hypothesis, $\exists$ has a winning strategy from position
$(A, x')$ if and only if $M, x' \models A$ for all subformulas $A$ of
$\lbrace A_1, \dots, A_n \rbrace$. These winning strategies can be
extended to provide a winning strategy from $\hearts(A_1, \dots,
A_n)$ by stipulating that $\exists$ move to $(\hearts(A_1, \dots,
A_n), (\lsem A_1 \rsem_M, \dots, \lsem A_n \rsem_M)$.
\eat{
Assume first that $\model,x \models \hearts (A_1,\ldots,A_n)$. At
position $(\hearts (A_1,\ldots,A_n),x)$ in $\MC_\Gamma(\model)$, let
$\exists$'s move to be to the position $(\hearts
(A_1,\ldots,A_n),(\lsem A_1 \rsem_M,\ldots,\lsem A_n \rsem_M))$.
Note that this is a legitimate move, as $M,x \models \hearts
(A_1,\ldots,A_n)$ gives $\gamma(x) \in \lsem \hearts \rsem_C (\lsem
A_1 \rsem_M, \dots,\lsem A_n \rsem_M)$. We now show that it is
impossible for $\forall$ to win if $\exists$ plays as above at
$(\hearts (A_1,\ldots,A_n),x)$. Assume that, at position $(\hearts
(A_1,\ldots,A_n),(\lsem A_1 \rsem_M,\ldots,\lsem A_n \rsem_M))$,
$\forall$ chooses $(A_i,x')$ with $x' \in \lsem A_i \rsem_M$. (If no
move is possible, then $\forall$ loses immediately.) Since $M,x'
\models A_i$, it follows by the induction hypothesis that $\exists$
has a winning strategy at position $(A_i,x')$. Now it is easy to see
that this strategy will also guarantee that $\exists$ wins at
position $(\hearts (A_1,\ldots,A_n),x)$.
}
Now assume that $\exists$ has a winning strategy from position
$(\hearts (A_1,\ldots,A_n),x)$ in $\MC_\Gamma(\model)$ under which
$\exists$ moves to position $(\hearts(A_1, \dots, A_n), (U_1, \dots, U_n))$
from position $(\hearts(A_1, \dots, A_n),x)$. By induction
hypothesis, we have that $x_i \models A_i$ for all $x_i \in U_i$ so
that $U_i \subseteq \lsem A_i \rsem_M$ and hence $\gamma(x) 
\in \lsem \hearts \rsem (\lsem A_1 \rsem_M, \dots, \lsem A_n
\rsem_M)$ by monotonicity of $\lsem \hearts \rsem$ whence $x \models
\hearts(A_1, \dots, A_n)$.
\eat{
This strategy provides sets $U_1,\ldots,U_n \subseteq C$ such that
$\gamma(x) \in \lsem \hearts \rsem_{C}(U_1,\ldots,U_n)$ when. Moreover,
for each $x_i \in U_i$, $\exists$ has a winning strategy at
$(A_i,x_i)$. Hence, by the induction hypothesis, $x_i \models A_i$
for all $x_i \in U_i$, and thus $U_i \subseteq \lsem A_i \rsem_M$.
The monotonicity of $\lsem \hearts \rsem$ now gives $\gamma(x) \in
\lsem \hearts \rsem_X(\lsem A_1 \rsem_M, \ldots, \lsem A_n
\rsem_M)$, and hence $M,x \models \hearts (A_1,\ldots,A_n)$.
}
\end{proof}
%The model checking game is used to show completeness of associated tableau calculi.
%------------------------------------------------------------------------- 
%------------------------------------------------------------------------- 
\section{Tableaux for the coalgebraic $\mu$-calculus}

\noindent
In this section, we characterise \emph{satisfiability} in terms of
non-existence of \emph{closed tableaux}. Given that our approach is
parametric both in the model class over which we interpret formulas
(embodied by the endofunctor) and the modal operators 
(given by the similarity type) that we use,
our tableau system will be parametric in a set of modal tableau
rules. Our tableaux will be constructed by applying the standard
rules for deconstructing propositional connectives, the modal rules
that are supplied as a parameter, and unfolding of fixpoints. To
ensure soundness and completeness of the ensuing calculus, we need
to ensure two properties:
\begin{enumerate}[(1)]
\item the supplied set of modal rules has to describe the model
class in a sound and complete way
\item topmost least fixpoints are only unfolded finitely often.
\end{enumerate}

\noindent For the first property, we introduce \emph{coherence conditions}
between the proof rules and the semantics that will guarantee
completeness. For the second property, we need to consider
\emph{traces} of formulas along the paths of the tableau and again
use a parity function to determine whether outermost $\mu$s are
unfolded only finitely often. As the unfolding of fixpoints may
create infinite branches, we conceptualise a tableau as a graph. A
closed tableau is then constructed according
to the given rules so that outermost 
least fixpoints are unfolded infinitely many times along any
path through the tableau.

We begin by describing the coherence conditions that will guarantee
soundness and completeness of the modal rules. These rules describe
the relationship between states and (coalgebraic) successors,
are of a particularly simple form, and are formulated in terms of
sequents.
\begin{defi}
A \emph{$\Lambda$-tableau sequent}, or just \emph{sequent}, is a
finite set of
$\Lambda$-formulas. We write $\Seq(\Lambda)$ for the set of
$\Lambda$-sequents. If $\Gamma \in \Seq(\Lambda)$ we write
$\Seq(\Gamma) = \lbrace \Delta \in \Seq(\Lambda) \mid \Delta
\subseteq \Cl(\Gamma) \rbrace$ for the set of sequents over the
closure of $\Gamma$. 

We identify a formula $A \in \FoRm(\Lambda)$
with the singleton set $\lbrace A \rbrace$, and write $\Gamma ;
\Delta = \Gamma \cup \Delta$ for the union of $\Gamma,
\Delta \in \Seq(\Lambda)$ as before. Substitution extends to sequents via
$\Gamma \sigma = \lbrace A \sigma \mid A \in \Gamma \rbrace$.
A \emph{monotone one-step tableau rule} for a similarity type
$\Lambda$ is of the form
\[ \frac{\Gamma_0}{\Gamma_1 \quad \dots
\quad \Gamma_n} \] 
where  $\Gamma_0 \in (\Lambda \cup
\bar\Lambda)(V)$  and 
$\Gamma_1, \dots, \Gamma_n \subseteq V$ for some set  $V \subseteq \Var$  
of propositional variables, every
propositional variable occurs at most once in $\Gamma_0$ and all
variables occurring in one of the $\Gamma_i$'s ($i > 0$) also occur in
$\Gamma_0$.
\end{defi}
\noindent
Monotone
tableau rules do not contain negated propositional variables, which
are not needed to axiomatise (the class of models induced by) monotone 
$\Lambda$ structures. The restriction on occurrences of
propositional variables is unproblematic, as variables that occur in
a conclusion but not in the premise and multiple occurrences of
variables in the premise can always be eliminated.
The set of one-step tableau rules is the only parameter in the
construction of tableaux for coalgebraic fixpoint logics.  The
coherence conditions relate rule sets with the interpretation
of modal operators purely on the level of properties of states
(subsets of a set $X$) and properties of successors (subsets of
$TX$).

\begin{defi}
Let $V \subseteq \Var$ be a set of propositional variables.
The \emph{interpretation} of a propositional sequent $\Gamma
\subseteq V \cup \bar V$ with respect to a set $X$ and a valuation
$\tau: V \to \Pow(X)$ is given by
$\lsem \Gamma \rsem_{X, \tau} = \bigcap \lbrace  \tau(p) \mid p \in
\Gamma \rbrace$, and
the
interpretation $\lsem \Gamma \rsem_{TX, \tau} \subseteq TX$ 
of a modalised sequent 
$\Gamma \subseteq 
(\Lambda \cup
\bar \Lambda)(V) $ is 
\[  
\lsem \Gamma \rsem_{TX, \tau} = \bigcap \lbrace \lsem \hearts \rsem_X
(\tau(p_1), \dots, \tau(p_n)) \mid \hearts (p_1, \dots, p_n) \in
\Gamma \rbrace. \]
If $T$ is a $\Lambda$-structure, then a set $\Rules$ of monotone tableau
rules for $\Lambda$ is \emph{one-step tableau complete} (resp.
\emph{sound}) with respect
to $T$ if
$\lsem \Gamma \rsem_{TX, \tau} \neq \emptyset$ if (only if) for all
$\Gamma_0 / \Gamma_1, \dots, \Gamma_n \in \Rules$ and all $\sigma: V
\to V$ with $\Gamma_0 \sigma \subseteq \Gamma$, there exists $1 \leq
i \leq n$ such that $\lsem \Gamma_i \sigma \rsem_{X, \tau} \neq
\emptyset$,
whenever 
$\Gamma \subseteq (\Lambda \cup \bar\Lambda)(V)$ and $\tau: V \to
\Pow(X)$.
\end{defi}
\noindent
Informally speaking, a set $\Rules$ of one-step tableau rules is
one-step tableau complete if 
a modalised sequent $\Gamma$ is satisfiable whenever 
a rule that matches $\Gamma$ 
has a satisfiable conclusion.
Some care has to be taken to ensure monotonicity of one-step rules
in concrete examples, in particular for the graded and the
probabilistic $\mu$-calculus. In order to obtain monotone rules for
these logics, we need to insist that rule conclusions only contain
prime implicants to avoid non-monotone occurrences of propositional
variables. This ensures that we avoid a (non-monotone) conclusion
consisting of e.g. $\Gamma; p$ and $\Gamma; \bar p$.
\begin{defi}
Suppose $I$ is a finite (index) set.
A \emph{prime implicant} of a boolean function $f: \lbrace 0, 1
\rbrace^I \to \lbrace 0, 1 \rbrace$ is a partial valuation $p: I \pto
\lbrace 0, 1 \rbrace$ with minimal domain of definition so that $f$
evaluates to $1$ under all total extensions of $p$. Given a family
$(p_i)_{i \in I}$ of propositional variables, every partial
valuation $p: I \pto \lbrace 0, 1 \rbrace$ (and hence every prime
implicant) induces a sequent
\[ \Gamma_p = \lbrace p_i \mid p(i) = 1 \rbrace \cup \lbrace \bar
p_i \mid p(i) = 0 \rbrace. \]

Now consider $k \in \Z$, a family $(r_i)_{i \in I}$ of integers and a family
$(p_i)_{i \in I}$ of propositional variables over the same index
set. For $I = I_0 \cup I_1$, we let
\[ \sum_{i \in I_0} r_i p_i +  \sum_{i \in I_1} r_i \bar{p_i} < k = \lbrace \Gamma_p
\mid p \mbox{ prime implicant of } f \rbrace \]
for the set of sequents induced by the prime implicants of the boolean function
$f: \lbrace 0, 1 \rbrace^I \to \lbrace 0, 1 \rbrace$ 
defined by $f(v) = 1 \iff \sum_{i \in I_0} r_i v(i) +  \sum_{i \in I_1} r_i (1-v(i)) < k$.
\end{defi}

In other words, the set of prime implicants of a boolean function
corresponds to the reduced disjunctive normal form of the associated
propositional formula. 
The notation $\sum_i r_i p_i < k$ introduced above allows us to read
a linear inequality involving propositional variables as a set of
sequents (that we will later use as the conclusion of a one-step
rule). If we think of the propositional variables $p_i$ as denoting
subsets $U_i$ of some set $X$, then the set of all points $x \in X$ that
satisfy the inequality $\sum_i \cf_{U_i}(x) \leq k$ is precisely the
set of points that satisfies the induced collection of sequents. (We
write $\cf_U: X \to \lbrace 0, 1 \rbrace$ for the characteristic
function of a subset $U \subseteq X$).
For one-step rules formulated in terms of linear inequalities, we
need this property to establish completeness.
\begin{lem} \label{lemma:cf}
Suppose $X$ is a set and $\tau: \Var \to \Pow(X)$ is a valuation of
propositional variables. 
Then $x$ satisfies one of the elements of 
$\sum_{i \in I_0} r_i p_i +  \sum_{i \in I_1} r_i \bar{p_i} < k$ iff $\sum_{i \in I_0}^n r_i \cf_{\tau(p_i)}(x) +  \sum_{i \in I_1} r_i \cf_{X \setminus \tau(p_i)}(x) <k$. That is,
\[ \sum_{i \in I_0} r_i \cf_{\tau(p_i)}(x) +  \sum_{i \in I_1} r_i \cf_{X \setminus \tau(p_i)}(x) <k \iff x \in \bigcup \lbrace
\lsem \Gamma \rsem_{(X, \tau)} \mid \Gamma \in 
\sum_{i\in I_0} r_i p_i + \sum_{i\in I_0} r_i \bar{p_i} < k  \rbrace
\]
for all $x \in X$, whenever  $r_1, \dots, r_n, k \in \Z$.
\end{lem}

\begin{proof}
First suppose that $x \in  \bigcup \lbrace
\lsem \Gamma \rsem_{(X, \tau)} \mid \Gamma \in 
\sum_{i \in I_0}  r_i p_i + \sum_{i \in I_1}  r_i \bar{p_i} < k  \rbrace$. Then there exists a prime 
implicant $p: I \pto
\lbrace 0, 1 \rbrace$ of the function $f$ given by $f(v) = 1 \iff
\sum_{i \in I_0} r_i v(i) + \sum_{i \in I_1} r_i (1 - v(i)) < k$ such that $x \in \lsem \Gamma_p \rsem_{(X,
\tau)}$.
Then the function $c: I_0 \cup I_1 \to \{0,1\}$ given by
\[      c(i)  \mathrel{:=} \left\{\begin{array}{l} 
	1 \; \mbox{if} \; x \in \tau(p_i) \; \mbox{and} \; i \in I_0 \\
	1 \; \mbox{if} \; x \not\in \tau(p_i) \; \mbox{and}  \; i \in I_1 \\
	0 \; \mbox{otherwise.}
\end{array} \right. \]
%$c(i) = 1 \iff x \in \tau(p_i)$ 
extends $p$ and
therefore $f(c) = 1$ whence 
\[ \sum_{i \in I_0} r_i \cf_{\tau(p_i)}(x) +  \sum_{i \in I_1} r_i \cf_{X \setminus \tau(p_i)}(x) = \sum_{
i \in I_0}
r_i c(i) + \sum_{i \in I_1} r_i c(i) < k.\] Now suppose that $ \sum_{i \in I_0} r_i \cf_{\tau(p_i)}(x) +  \sum_{i \in I_1} r_i \cf_{X \setminus \tau(p_i)}(x) < k$
and consider the valuation 
\[ v(i) = \left\{ \begin{array}{l} 1 \mbox{ if } x \in \tau(p_i), i \in I_0 \\
 1 \mbox{ if } x \not\in \tau(p_i), i \in I_1\\
 0 \mbox{ otherwise.} 
 \end{array} \right. \]
We
have that $f(v) = 1$ and therefore obtain a prime implicant $p: I
\pto \lbrace 0, 1 \rbrace$ of $f$ such that $v$ extends $p$ and $x
\in \lsem \Gamma_p \rsem_{(X, \tau)}$.
\end{proof}
%
%Example rules
\noindent
This finishes our discussion of prime implicants and we are ready
to have a look at several examples.
We use the following one-step rules to axiomatise the model classes
introduced in Example \ref{example:logics}.

\noindent
\begin{exa}\label{example:rules}
%The following rule sets are used for the logics introduced in 
%Example \ref{example:logics},
\begin{enumerate}[(1)]
\item
The (standard) modal logic of Kripke frames is axiomatised by all the
instances of
\[
  (K) \frac{\Diamond p_0; \Box p_1; \dots; \Box p_n}{p_0; p_1; \dots; p_n}
\]
where $n \geq 0$.
\item the set of one-step rules associated with graded modal logic
(and the graded $\mu$-calculus, interpreted over finitely branching
multigraphs) can be axiomatised by the rule schema
\[
	(G)  \frac{\langle k_1 \rangle p_1; \dots; \langle k_n \rangle p_n;
[ l_1 ] q_1; \dots; [l_m] q_m}
       {\sum_{j = 1}^m s_j \bar{q_j}  - \sum_{i = 1}^n r_i p_i < 0} 
			 \]
where $m, n \geq 0$ and $r_i, s_j \in \Nat \setminus \lbrace 0 \rbrace$ and 
$\sum_{i = 1}^n r_i (k_i + 1) \geq 1 + \sum_{j = 1}^m s_j l_j$.
\item The set of rules associated to the probabilistic
$\mu$-calculus comprises all instances of
\[
  (P) \frac{\langle a_1 \rangle  p_1; \dots ; \langle a_n \rangle  p_n; [b_1] q_1 ;
\dots ; [b_m] q_m}
{ \sum_{j = 1}^m s_j \bar q_j - \sum_{i = 1}^n r_i p_i < k} \quad
\]
where $m, n \geq 0$, $r_i, s_j \in \Nat \setminus \lbrace 0 \rbrace$ and
$\sum_{i = 1}^n r_i a_i - \sum_{j = 1}^m s_j b_j \leq k$  if
$n > 0$ and $ - \sum_{j = 1}^m s_j b_j < k$  if $n = 0$.
\item 
For coalition logic, we have all instances of
\[
(C_1)\frac{[C_1] p_1; \dots; [C_n] p_n}{p_1; \dots ; p_n} 
\qquad
   (C_2)\frac{[C_1] p_1; \dots; [C_n] p_n; \bar{[D]} q; \bar{[N]} r_1;
	 \dots ; \bar{[N]} r_m}
	 {p_1; \dots; p_n; q; r_1; \dots; r_m} 
\]
where again $m, n \geq 0$. Both rules are subject to the side
condition that the $C_i$ are disjoint. For $(C_2)$ we moreover
require
$C_i \subseteq D$.
\item
  Finally, the rule set associated to monotone modal logic contains
	the single rule
\[
	(M)  \frac{\Box p; \Diamond q}{p; q}.
\]
\end{enumerate}
\end{exa}
%\begin{lemma}  \label{lemma:g-p-monotone}
%All instances of $(G)$ and $(P)$ are monotone.
%\end{lemma}
%
In the rule schemas $(G)$ and $(P)$, we note 
that $\sum_i r_i a_i < k$ is a set of (propositional) sequents, and therefore
qualifies as the conclusion of a tableau rule. To ensure
monotonicity, we have to ensure that no literal appears negatively.
This is a direct consequence of the following: 

\begin{lem} \label{lemma:rule-monotone}
Suppose that $p$ is a prime implicant of the boolean function $f:
\lbrace 0, 1 \rbrace^I \to \lbrace 0, 1 \rbrace$ given by $f(v) = 1$
iff $\sum_{i \in I} r_i v_i < k$, where $(r_i)_{i \in I}$ is a
sequence of nonzero integers. Then
$p(i) = 1$ or undefined whenever $r_i < 0$ and analogously, $p(i) =
0$ or undefined whenever $r_i > 0$. In particular, all instances of
$(G)$ and $(P)$ are monotone.
\end{lem}

\begin{proof}
We only demonstrate the first item, the second is analogous.
Suppose, for a contradiction, that $p(i) = 0$ and $r_i < 0$. Then,
removing $i$ from the domain of definition of $p$ yields a function
$q: I \pto \lbrace 0, 1 \rbrace$ such that all total extensions $e$
of $q$ still satisfy $f(e) = 1$, contradicting the minimality of
$p$.
\end{proof}

It is easy to see that every $\Lambda$-structure admits a one-step
sound and complete set of one-step tableau rules. While this
demonstrates that our approach is applicable to all conceivable
$\Lambda$-structures, the challenge of finding a tractable
representation of the rule set remains, which is crucial for a
complexity analysis.
An adaptation of \cite[Theorem 17]{Schroder:2006:FMC} to the setting
of monotone tableau rules shows that one-step complete rule sets
always exist.

\begin{prop} \label{propn:mon-existence}
Every monotone $\Lambda$-structure admits a one-step tableau sound and one-step tableau complete set of monotone tableau rules.
\end{prop}
\begin{proof}
Suppose that $T$ is a monotone $\Lambda$-structure. We show that
there exists a set $\Rules$ of monotone tableau rules so that
$\Rules$ is one-step tableau sound and one-step tableau complete for $T$, essentially by
showing that the set of all monotone one-step sound rules is indeed one-step
complete. We let $\Rules$ consist of all monotone tableau rules
$\Gamma_0 / \Gamma_1, \dots, \Gamma_n$ that satisfy
\[ \lsem \Gamma_1 \rsem_{X, \tau} = \dots = \lsem \Gamma_n \rsem_{X,
\tau} = \emptyset \implies \lsem \Gamma_0 \rsem_{TX, \tau} =
\emptyset \]
for all sets $X$ and valuations $\tau: V \to \Pow(X)$. We claim that
$\Rules$ is one-step tableau sound and one-step tableau complete.

First, for one-step tableau soundness, 
suppose that $\tau: V \to \Pow(X)$ is given and $\lsem \Gamma
\rsem_{TX, \tau} \neq \emptyset$ for some $\Gamma \subseteq (\Lambda
\cup \bar\Lambda)(V)$. For $\Gamma_0 / \Gamma_1,
\dots, \Gamma_n \in \Rules$ and a renaming $\sigma: V \to V$ such
that $\Gamma_0 \sigma \subseteq \Gamma$, we have to show that 
$\lsem \Gamma_i \sigma \rsem_{X, \tau} \neq \emptyset$ for some $1
\leq i \leq n$. Assume, for a contradiction, that $\lsem \Gamma_i
\sigma\rsem_{X, \tau} = \emptyset$ for all $1 \leq i \leq n$. Then, for
$\tau'(p) = \tau(\sigma(p))$ we have $\lsem \Gamma_i \rsem_{X,
\tau'} = \emptyset$ for all $1 \leq i \leq n$ so that
$\lsem \Gamma_0 \sigma \rsem_{TX, \tau} = \lsem \Gamma_0 \rsem_{TX,
\tau'} = \emptyset$, contradicting $\lsem \Gamma_0 \sigma \rsem_{TX, \tau}
\supseteq \lsem \Gamma \rsem_{TX, \tau} \neq \emptyset$.

For one-step tableau completeness, we directly show the contrapositive. 
Assume that $\lsem \Gamma \rsem_{TX, \tau} = \emptyset$ for some set
$X$ and some valuation $\tau: V \to \Pow(X)$. We show that, in this
case, there exists $\Gamma_0 / \Gamma_1 \dots, \Gamma_n \in \Rules$
and $\sigma: V \to V$ such that $\Gamma_0 \sigma \subseteq \Gamma$
and $\lsem \Gamma_i \sigma \rsem_{X, \tau} = \emptyset$.

So suppose that $\lsem \Gamma \rsem_{TX, \tau} = \emptyset$ and
consider the set
\[ S = \lbrace \Delta \subseteq V_\Gamma \mid \lsem \Delta \rsem_{X,
\tau} = \emptyset \rbrace \]
where $V_\Gamma$ denotes the set of propositional variables occurring
in $\Gamma$.
If we let $S = \lbrace \Gamma_1, \dots, \Gamma_n \rbrace$, it
suffices to show that $\Gamma / \Gamma_1, \dots, \Gamma_n \in
\Rules$. So suppose $\rho: V \to \Pow(Y)$ is a valuation such that
$\lsem \Gamma_1 \rsem_{Y, \rho}  = \dots = \lsem \Gamma_n \rsem_{Y,
\rho} = \emptyset$. We show that $\lsem \Gamma \rsem_{TY, \rho} =
\emptyset$. To this effect, we claim that there exists a function
$f: Y \to X$ such that $y \in \rho(p) \implies f(y) \in \tau(p)$ for
all $p \in V_\Gamma$. For if not, there exists $y \in Y$ for which a
suitable $f(y)$ cannot be found, i.e.~for all $x \in X$ we may find $p_x \in
V_\Gamma$ such that $x \notin \tau(p_x)$ but $y \in \rho(p_x)$. For
the sequent $\Delta = \lbrace p_x \mid x \in X \rbrace$ we then
obtain $\lsem \Delta \rsem_{X, \tau} = \emptyset$ whence $\Delta \in
S$ but $y \in \lsem \Delta \rsem_{Y, \rho}$, contradicting $\lsem
\Gamma_i \rsem_{Y, \rho} = \emptyset$ for all $i = 1, \dots, n$.

By construction, the function $f$ satisfies $\rho(p) \subseteq
f\inv(\tau(p))$ for all $p \in V_\Gamma$, which gives, by
monotonicity of the $\Lambda$-structure $T$, that
\[ \lsem \Gamma \rsem_{TY, \rho} \subseteq \lsem \Gamma \rsem_{TY,
f\inv \circ \tau} = (Tf)\inv (\lsem \Gamma \rsem_{TX, \tau}) =
\emptyset \]
as required, where the second equality is by naturality of predicate liftings.
\end{proof}

\noindent
In the examples, we can find concrete (and tractable)
representations of one-step complete rule sets.
\begin{prop} \label{propn:one-step-tab-complete}
The rule sets introduced in Example~\ref{example:rules} are both one-step
tableau sound and one-step tableau complete with respect to the corresponding structures
defined in Example~\ref{example:logics}.
\end{prop}
\begin{proof}
It is straightforward to see that a set of monotone rules is one-step tableau
complete iff the set of proof rules arising by negating and swapping premise
and conclusion is one-step sound and strictly one-step complete  in
the sense of \cite{Schroder:2008:PBR}, where soundness and
completeness is established for the dual rule sets. The case of
graded and probabilistic modal logic additionally requires to invoke
Lemma~\ref{lemma:cf} together with Lemma 3.18 of \emph{op.~cit.}
\end{proof}

\noindent
We now introduce the set of tableau rules that we are using to
axiomatise the coalgebraic $\mu$-calculus. As to be expected, these
rules are parametric in a set of one-step rules, and we will
instantiate our results to the logics introduced in Example
\ref{example:logics} with help of the previous proposition. Along
with the tableau rules, we also introduce rule blueprints and rule
representations that will aid us in the definition of paths through
a tableau later on.
\begin{defi} \label{defn:rules}
The set $\Tab\Rules$ of \emph{tableau
rules} induced by a set $\Rules$ of one-step rules contains the
propositional and fixpoint rules, the modal rules $(\mod)$ and the
axiom (rule) below:
\[ (\land) \frac{\Gamma; A \land B}{\Gamma; A; B}  
\quad (\lor) \frac{\Gamma; A \lor B}{\Gamma; A \quad \Gamma; B} 
\quad
(\fix) \frac{\Gamma; \eta p. A}{\Gamma; A [p := \eta p.A]} 
\quad
(\mod) \frac{\Gamma_0 \sigma, \Delta}{\Gamma_1 \sigma \dots \Gamma_n
\sigma} 
\quad
(\mathsf{Ax}) \frac{\Gamma, A, \bar A}{} \]
Here $\Gamma_0 / \Gamma_1 \dots \Gamma_n \in \Rules$ and $\sigma:
\Var \to \FoRm(\Lambda)$ is a substitution satisfying $\sharp(
\Gamma_0) = \sharp(\Gamma_0 \sigma)$ where $\sharp$ denotes
cardinality.  The formulas $A \land B$, $A \lor B$ and $\eta p. A$ are
called \emph{principal} in the rules $(\land)$, $(\lor)$ and
$(\fix)$.
A \emph{rule blueprint}  is of the form $A \land
B$, $A \lor B$, $\eta p. A$, $(A, \bar A)$ or $(r, \sigma)$, where $r \in \Rules$ and $\sigma: V_0 \to
\FoRm(\Lambda)$ is a substitution satisfying $\sharp(
\Gamma_0) = \sharp(\Gamma_0 \sigma)$ and 
$V_0 \subseteq \Var$ is the set of variables occurring
in $r$. We write $\Bp(\Rules)$ for the
set of rule blueprints over the set $\Rules$ of one-step rules.
A \emph{rule representation} is a tuple $(\Gamma, \flat)$ where
$\Gamma \in \Seq(\Lambda)$ and $\flat$ is a rule blueprint that satisfies
\begin{enumerate}[$\bullet$]
\item $\flat \in \Gamma$ if $\flat$ is of the form $A \land B$, $A
\lor B$ or $\eta p. A$
\item $A,\bar A \in \Gamma$ if $\flat = (A,\bar A)$
\item $\Gamma_0 \sigma \subseteq \Gamma$ if $\flat = (r, \sigma)$
and $r = \Gamma_0 / \Gamma_1 \dots \Gamma_n$.
\end{enumerate}
Each rule representation $(\Gamma, \flat)$ induces a tableau rule
$\rho(\Gamma, \flat) \in \Tab\Rules$ given by
\begin{align*}
  \rho(\Gamma, A \land B) & = \frac{\Gamma}{A, B, \Gamma'} &
	\rho(\Gamma, A \lor B) & = \frac{\Gamma}{A, \Gamma' \quad B,
	\Gamma'} \\
	\rho(\Gamma, \eta p. A) & = \frac{\Gamma}{A[p := \eta p. A],
	\Gamma'} &
	\rho(\Gamma, (r, \sigma)) & = \frac{\Gamma}{\Gamma_1 \sigma \dots
	\Gamma_n \sigma}\\
        \rho(\Gamma, (A,\bar A)) & = \frac{\quad \Gamma \quad}{}
\end{align*}
where $\Gamma' = \Gamma \setminus \lbrace \flat \rbrace$ in the
first three clauses, and $r
= \Gamma_0 / \Gamma_1 \dots \Gamma_n$ in the fourth clause.
\end{defi}

\noindent
The restriction $\sharp(\Gamma_0 \sigma) = \sharp(\Gamma_0)$ on
instances of one-step rules ensures that the substitution does not
identify literals in the premise of a one-step rule, which implies
that only finitely many modal rules are applicable to any sequent.
Similarly, because of the restriction $\sharp(\Gamma_0 \sigma) = \sharp(\Gamma_0)$ on substitutions 
and on the size of the domain of such substitutions in rule representations, it is also easy to see that for any $\Gamma \in \Seq(\Lambda)$ the set of rule representations $(\Gamma,\flat)$ is finite.
This will enable us to deduce decidability, and indeed complexity
bounds later. We will, however, need to require that 
the set of modal rules is contraction closed in order to ensure completeness of the restricted calculus.

\begin{rem}\label{rem:contra_closed}
      Alternatively, we could also prove soundness and completeness for the tableau calculus
      without the restriction $\sharp(\Gamma_0 \sigma) = \sharp(\Gamma_0)$ and without
      requiring contraction closure for the set of modal rules. In this case, in order to obtain
       decidability, we would have to require  contraction-closedness of the modal rules.
      This is essential for proving that we can restrict the calculus to 
      (finitely many) instances $(r,\sigma)$ of modal rules
      with non-identifying substitutions $\sigma$.
\end{rem}

Our definition of rule blueprints and rule representations may seem
a bit bureaucratic at first sight, so some comments are in order. If
we understand a tableau as a two-player game where $\forall$ plays a
tableau rule and $\exists$ selects a conclusion, the winning
condition for $\exists$ stipulates that least fixpoints are not
unfolded infinitely often. This condition is formalised in terms of
the evolution of formulas along a path in a tableau, which in turn
necessitates that we can re-construct the rules applied to tableau
nodes. This is achieved by annotating each tableau node with a rule
blueprint. Together with the node label, the blueprint forms a rule
representation which in turn induces a rule. We use this mechanism
for two reasons:
\begin{enumerate}[$\bullet$]
\item for propositional rules and the fixpoint rule, the rule
blueprint records the principal formula, that we need to track to
define traces later. Moreover, we can distinguish between the different
conclusions of the induced rule, and
\item for modal rules, the rule blueprint is an unsubstituted
one-step rule, which allows us to track (unsubstituted)
propositional variables, which is again needed for the definition of
traces.
\end{enumerate}
The usefulness of the blueprints and rule representations will become clearer
in Definition~\ref{def:trace} where we define the set of traces through a tableau path.
We are now ready to introduce the notion of tableau that we will use
throughout the paper. 
As fixpoint rules generate infinite paths,
we formalise tableaux as finite, rooted graphs. As a consequence, 
\emph{closed} tableaux are  finitely represented proofs of the
unsatisfiability of the root formula.
\begin{defi}
A \emph{tableau} for a clean, guarded sequent $\Gamma \in \Seq(\Lambda)$ is a
finite, directed, rooted and labelled graph $(N, K, R, \ell, \alpha)$ where $N$ is
the set of nodes, $K \subseteq N \times N$ is the set of edges, $R$
is the root node and $\ell: N \to \Seq(\Gamma)$ is a labelling
function such that $\ell(R) =   \Gamma$ and 
$\alpha: N \pto \Bp(\Rules)$ is a partial function (that we think of
as an annotation) satisfying
\begin{sparitemize}
\item $\alpha(n)$ is defined iff there exists a tableau rule with premise
$\ell(n)$ iff $K(n) \neq \emptyset$.
\item if $\Gamma_0 / \Gamma_1 \dots \Gamma_n  = \rho
(\ell(n), \alpha(n))$ then $\lbrace \Gamma_1, \dots,
\Gamma_n \rbrace = \lbrace \ell(n') \mid n' \in K(n) \rbrace$
\end{sparitemize}
where  $K(n) = \lbrace n' \mid (n, n') \in K \rbrace$ and $\rho$ is
as in Definition \ref{defn:rules}.
\end{defi}
In other words, tableaux are sequent-labelled graphs where
a rule has to be applied at a node if the node label matches a rule
premise, and no rule may be applied otherwise.  The purpose of the
annotation $\alpha$ is to record which rule (if any) has been
applied at a particular node. To keep track of whether least
fixpoints are unfolded infinitely often, we record the 
\emph{unsubstituted} one-step rule (together with a substitution) at
modal nodes, as we need to track the evolution of formulas along
one-step rules, where propositional variables may become identified
by a substitution. Moreover,
it may be the case that
two different one-step rules generate the same rule instance: both
rules
$\hearts p / p$ and $\hearts p, \hearts q / p$ generate the
instance $\hearts A / A$. As we will be required to
traverse infinite loops in a tableaux to ensure that only greatest
fixponits are unfolded infinitely often, we need to ensure that the
identity of a rule does not change when nodes are encountered
multiple times. 

The reader might wonder why we make a distinction between 
nodes in a tableau and their labels. The technical reason for this
is that we need to run an automaton in parallel to the tableau, so
that the same sequent may be associated with different automata
states (see the definition of the tableau game in 
Section~\ref{sec:tabgame}). Informally speaking, we have to allow
for enough paths through a tableau to ensure completeness.
We can view a tableau as a strategy of $\forall$
in this tableau game, where  $\forall$ tries to prove that a given sequent is
not satisfiable. Accordingly, a closed tableau will correspond to a winning strategy for him
in the tableau game. An identification of nodes and sequents in a tableau
would mean that the corresponding strategy of $\forall$ in the tableau game would only depend
on the set of formulas with which a position of $\forall$ is labeled. 
We cannot guarantee, however, that $\forall$ has a winning strategy
of this special kind, even if he has some winning strategy. Therefore, in order to 
be able to represent arbitrary strategies of $\forall$ in the tableau game as
tableaux,  
we have to have the possibility to distinguish between nodes and their labels.
The only restriction we make is that the tableau graph is finite, i.e.~we only consider
strategies of $\forall$ with bounded memory.

Our goal is to show that a formula $A \in \FoRm(\Lambda)$ is
satisfiable iff no tableau for $A$ ever closes. In a setting without
fixpoints, a tableau is closed iff all leaves are labelled with
axioms.  Here we also need to consider 
infinite paths, and ensure that
only greatest fixpoints are unfolded infinitely often at the top
level of an infinite path. As in \cite{Niwinski:1996:GMC}, this
necessitates to consider the set of \emph{traces} through a given
tableau. Informally, a trace records the evolution (by application
of tableau rules) of a single formula through a tableau. Formally,
we associate binary relations with tableau rules, and traces arise
by sequencing these relations.
\begin{defi} \label{def:trace}\label{defn:tiles}
Suppose that $\Tab = (N, K, R, \ell, \alpha)$ is a 	tableau for
$\Gamma$. A
\emph{path} through $\Tab$ is a finite or infinite sequence
\[ \pi: n_0 \stackrel{c_0}\to n_1 \stackrel{c_1}\to n_2 \dots \]
where $n_0 = R$, $n_{i+1} \in K(n_i)$ and $c_i \in \Nat$ satisfying
that $\ell(n_{i+1})$ is the $c_i$-th conclusion of the rule
represented by $(\ell(n_i), \alpha(n_i))$. A path is called {\em complete} if it is infinite or if it ends at a node $n \in N$ with $K[n] = \emptyset$.

A \emph{trace} through a path $\pi$ is a finite or infinite sequence of
formulas $(A_0, A_1, \dots)$ such that $A_i \in \ell(n_i)$ and
$(A_i, A_{i+1}) \in \Tr ( \ell(n_i), \alpha(n_i), c_i )$
where the relations $\Tr ( \Gamma, \flat, i ) \subseteq
\FoRm(\Lambda) \times \FoRm(\Lambda)$ are given as follows:
\begin{enumerate}[$\bullet$]
\item $\Tr(\Gamma, A_1 \land A_2, 1) = \lbrace (A_1 \land A_2, A_1),
(A_1 \land
A_2, A_2) \rbrace \cup \Diag(\Gamma \setminus \lbrace A \land B \rbrace)$
\item $\Tr(\Gamma, A_1 \lor A_2, i) = \lbrace (A_1 \lor A_2, A_i)
\rbrace \cup \Diag(\Gamma \setminus \lbrace A_1  \lor A_2 \rbrace)$
for $i = 1, 2$.
\item $\Tr(\Gamma, \eta p. A, 1) = \lbrace (\eta p. A, A [p := \eta
p. A]) \rbrace \cup \Diag(\Gamma \setminus \lbrace \eta p. A
\rbrace)$
\item $\Tr(\Gamma, (r, \sigma), i) = \lbrace (\hearts(p_1, \dots,
p_n)\sigma, p_j \sigma) \mid \hearts(p_1, \dots, p_n) \in \Gamma_0,
p_j \in \Gamma_i \rbrace$\\ where $r = \Gamma_0 / \Gamma_1 \dots
\Gamma_n$.
\end{enumerate}
Here $\Diag(X) =  \lbrace (x, x) \mid x \in X \rbrace$ is the
diagonal on a set $X$. The triples $(\Gamma,\flat, i)$ where $(\Gamma,\flat)$ is a
rule representation, $i \in \Nat$ and $\rho(\Gamma,\flat)$, the rule represented
by $(\Gamma,\flat)$, has at least $i$ conclusions, are called
\emph{trace tiles}.
Finally, a tableau $\Tab$ with root node labelled by $\Gamma$
is \emph{closed}, 
if the end node of all finite
paths through $\Tab$ of maximal length that starts in the root node
is labelled with a tableau axiom,
and every infinite path starting
in the root node carries at least one bad trace with respect to
a parity function $\Omega$ for $\Gamma$.
\end{defi}

%}
%
\noindent
Informally, a path through a tableau is a sequence of nodes,
together with the information which rule has been applied to nodes,
and we cannot have a path that ends in a node to which $(\mathsf{Ax})$ was
applied.
As for the construction of tableaux, the construction of traces
requires that we pick the same conclusion every time a node is
traversed. While in the instance $A \lor B, A, B / A, B$  of $(\lor)$,
both conclusions are identified, they are not equivalent from the
point of view of traces, as the `left' conclusion continues the
trace from $A \lor B$ to $A$ whereas the right conclusion continues
the same trace to $B$. This difficulty does not arise in
\cite{Niwinski:1996:GMC} where tableaux are formalised as
sibling-ordered trees, and the rule blueprints used here
serve essentially the same purpose. The traces through a path 
are calculated using the so-called trace tiles. A trace tile records which
rule has been applied in a node that is visited by the path and through which
of the successors of the node that path is continuing.
It should be noted that for $\Gamma \in \Seq(\Lambda)$ the set 
	\[ \Sigma_\Gamma = \{ (\Delta,\flat,i) \mid  (\Delta,\flat,i) \; \mbox{is a trace tile and } \Delta \in \Seq(\Gamma) \} \] 
is finite because, as remarked after 
Definition~\ref{defn:rules}, for each $\Delta \in \Seq(\Gamma)$ there are only finitely many 
rule representations $(\Gamma,\flat)$. We stress this fact, because later on we will use $\Sigma_\Gamma$ as alphabet of the parity automaton that is essential for the definition
of our tableau game. 

\begin{exa}
Assume that we have a tableau where
$(\lor)$ has been applied at the node labelled with $ A
\lor \mu p.B; C$ and $(\fix)$ has been applied at the node labelled
with $\mu p.  \Diamond B;C$. (We identify nodes and their labels
here for simplicity.)
Then the path
\[ A \lor \mu p.B; C \stackrel{2}{\longrightarrow} \mu p.
\Diamond B;C
\stackrel{1}{\longrightarrow} \Diamond B[p := \mu p. B];C \dots \]
supports the traces 
$(A \lor \mu p.\Diamond B, \mu p. \Diamond B, \Diamond B[p := \mu p.B], 
\dots)$  and 
$(C, C, C, \dots)$.
Note that there is
no trace on this path that starts with $A$.
\end{exa}
\noindent
We now continue the development of the general theory and first establish
soundness of the tableau calculus: satisfiable sequents cannot have
closed tableaux. This relies on 
Theorem \ref{thm:adequacy}, as a winning strategy for $\exists$ in
the model checking game can be used to construct a path through any
tableau that carries a bad trace.
\begin{thm}\label{thm:noclosedtab}
	Let $\Rules$ be a one-step tableau complete set of monotone rules for the modal similarity type $\Lambda$, and let $\Gamma \in \Seq(\Lambda)$ be clean and guarded. If $\Gamma$ is satisfiable in some
	model $M=(X,\gamma,h)$, then no closed tableau for $\Gamma$ exists.
\end{thm}
\begin{proof} 
Consider a model $\model=(X,\gamma,h)$ and $x \in X$ such that $\model, x
\models \Gamma$ and let
$\Tab=(N,K,R,\ell,\alpha)$ be a
tableau for $\Gamma$.  As $\model,x \models
\Gamma$, Theorem~\ref{thm:adequacy} implies that
$\exists$ has a history-free winning strategy $g$ in
$\MC_\Gamma=\MC_\Gamma(\model)$ from all positions $(B,x)$ of the game
board with $B \in  \Gamma$. 
We now establish that
there exists a complete path and an associated sequence of
model states satisfying the formulas on this path that can be
contracted to a play in the model checking game.
 More precisely, we
establish the existence of a path
$\pi=n_0 c_0 n_1 c_1 \ldots n_l c_l\ldots$ through $\Tab$
and a sequence $\chi = x_0 x_1 \ldots x_l \ldots$ of states that
satisfy
\begin{enumerate}[(i)] 
\item\label{path1} $n_0 = R$ and $x_0 = x$ and $\model,x_i \models \ell(n_i)$ whenever 
$n_i$ is defined,
\item\label{path2}
for each trace $\tau= B_0 B_1 \ldots B_l \ldots$ through
$\pi$ there exists a play $(A_0, y_0) (A_1, y_1) \dots$ (where we do not record
the positions that have subsets of the model as second component) 
that is
played according to $g$ and there is an increasing sequence $0 = s_0 < s_1 < 
\dots$ of indices such that
$B_{s_1} B_{s_2} \dots = A_0 A_1 \dots$ and $y_0 y_1 \dots = x_{s_0}
x_{s_1} \dots$ 
where  $B_i = B_{s_j}$
 and $x_i =
x_{s_j}$ 
whenever $s_j \leq i  < s_{j+1}$.
\end{enumerate} 
Once this claim is established, it follows that $\Tab$ cannot be
closed: 
%take an annotation $\alpha$ of $\Tab$ and 
consider the path $\pi$ just constructed.
If $\pi$ is finite, the label $\Delta$ of the last node of $\pi$
cannot be a tableau axiom, as $\Delta$ is satisfiable by
construction. In case $\pi$ is infinite, every trace $\tau$ through
$\pi$ induces a $\MC_\Gamma$-play that is played according to
$\exists$'s winning strategy $g$ which implies that $\tau$ is not
bad. Taken together, this shows that $\Tab$ cannot be closed, so it
remains to establish the claim.

We construct the required path $\pi$ and the
sequence $\chi$ of states in a step-by-step fashion, 
starting at the root of the
tableau and at the state $x$, i.e.~we put $n_0 =  R$ and $x_0
= x$.  
So suppose that a path $\pi=n_0 c_0
\ldots c_{j-1} n_j$ and a sequence of model states $x_0 \ldots
x_j$ satisfying (\ref{path1}) and (\ref{path2}) above have already
been constructed, and $\pi$ is not yet complete. 
We distinguish
cases on the rule $r=\rho(\ell(n_j), \alpha(n_j))$ 
applied at (the last) node $n_j$.

We begin with the case where $r=\Delta;D_1\vee D_2/
\Delta;D_1 \quad \Delta;D_2$ is an instance of the disjunction rule. 
In this case, we can find tableau
nodes $m_1$ and $m_2$ with $\ell(m_1) = \Delta, D_1$ and $\ell(m_2)
= \Delta, D_2$ and 
$K(n_j)\supseteq \{m_1,m_2\}$. 
Suppose that
$g(D_1\vee D_2,x_j)=(D_i,x_j)$ for $i\in \{1,2\}$.
We 
put $c_j = i$, $n_{j+1} =  m_i$ and $x_{j+1}
= x_j$. 
Then the extended path $\pi' = \pi
c_j n_{j+1}$ and $x_0 x_1 \ldots x_j x_{j+1}$ satisfy condition
(\ref{path1}) of our claim. Obviously we have  $x_{j+1} \models
\Delta$. Furthermore, $x_{j+1} \models D_i$ as $(D_i,x_{j+1})$ is a winning position
of $\exists$ in $\MC_\Gamma$. Thus, as $ \ell(n_{j+1}) = \Delta;D_i$, we have
$x_{j+1} \models \ell(n_{j+1})$ as required.
To see that (\ref{path2}) also holds,
consider a trace $\tau' = B_0
\ldots B_j B_{j+1}$ through $\pi'$ and let $P=(A_0,y_0) (A_1,y_1)
\ldots (B_j,y_k)$ be the partial play of $\MC_\Gamma$ that is associated
to $\tau=B_0 \ldots B_j$ and
that is played according to $g$.
If $B_j \neq  D_1 \vee D_2$ we have $B_j = B_{j+1}$ and $P$ can be
chosen as the corresponding $\MC_\Gamma$-play for $\tau'$. Otherwise,
if $B_j = D_1 \vee D_2$,  we have $B_{j+1} = D_i$
and we extend $P$ to $(A_0,y_0) (A_1,y_1) \ldots (D_1\vee
D_2,y_k)(D_i,y_{k+1})$ with $y_{k+1} = y_k$.  This $\MC_\Gamma$-play
now satisfies condition (\ref{path2}) of our claim.

The cases where $r$ is an instance of the conjunction or fixpoint
rules  are similar (even easier, as these rules only have one
conclusion).  
So suppose that $r$ is an instance of a modal rule.
That is, $r=\rho(\ell(n_j),\alpha(n_j))$ with
$\alpha(n_j) = (r,\sigma)$ for some rule $\Delta/\Delta_1,
\cdots,\Delta_s$ with $\Delta \sigma \subseteq \ell(n_j)$
%$\Delta;\Sigma/\Delta_1,\cdots,\Delta_s$ 
and $K(n_j)
\supseteq \{m_1,\ldots,m_s\}$ with
$\ell(m_i)=\Delta_i \sigma$ for $i \in \{1,\ldots,s\}$.
%with $\alpha(n_j) = \Sigma$ and
%$l(n_i') = \Delta_i$ for all $i \in \{1,\ldots,s\}$.  Then there
%exists
%$r'=\Delta'/\Delta_1',\cdots,\Delta_s' \in \Rules$ and a
%substitution $\eta$
%such that $r=\Delta'\eta;\Sigma /
%\Delta_1'\eta,\cdots,\Delta_s'\eta$.
We define a valuation
$\tau: V_{\Delta} \to \Pow(X)$ on the set
$V_{\Delta}$ of variables occurring in $\Delta$ by stipulating
that $\tau(p) = U_k$ where the (unique) occurrence of $p = p_k$
is in the formula $\hearts(p_1,\ldots,p_r) \in
\Delta$ and
$g(D,x_{j})=(D, (U_1,\ldots,U_r))$ with $D=\hearts(\sigma(p_1),
\ldots,\sigma(p_r))$.
As $g$ is winning for $\exists$ in $\MC_\Gamma$ at position
$(D',x_j)$ for all $D' \in \Delta\sigma$, it follows that
$\gamma(x_j) \in \lsem \Delta \rsem_{TX,\tau}$,
which implies that $\lsem \Delta \rsem_{T X ,\tau} \not=
\emptyset$. 

By one-step tableau completeness, 
$\lsem \Delta_i  \rsem_{X,\tau} \neq \emptyset$ for some $i \in
\{1,\ldots,s\}$. We now extend $\pi$ to a path $\pi' = \pi \,
m_i$ and let $x_{j+1}$
be an arbitrary element of $\lsem \Delta_i \rsem_{X,\tau}$.
Now consider a trace $\tau'$ through $\pi'$ that ends in some formula $A$
with $A = \sigma(p_A)$ for some $p_A \in \Delta_i$. 
Then, by Definition \ref{def:trace}, $\tau'$ is of
the form 
$\tau A$ where $\tau$ is a trace through $\pi$ ending in
a formula of the form $B=
\hearts(p_1, \dots, p_n) \sigma$, $\hearts(p_1, \dots, p_n) \in
\Delta$, and $p_A =p_k$ for some $k \in \{1,\ldots,n\}$.

By assumption, there exists a
corresponding $\MC_\Gamma$-play, played according to $g$, that ends
in position $(\hearts(p_1, \ldots,p_n)\sigma,x_j)$. 
This
play can now be extended by $\exists$ moving to 
$g(\hearts(p_1,\ldots, p_n)\sigma,x_j)$ $=
(\hearts(p_1,\ldots, p_n)\sigma, (U_1,\ldots U_n))$.
% with $U=\sigma(p)$.  
We extend this play
letting
$\forall$ move to $(A,x_{j+1})$.  The latter move is legitimate as
$\sigma(p_k) = A$ and because
$x_{j+1} \in \lsem \Delta_i  \rsem_{X,\tau} = \bigcap_{p \in
\Delta_i} \tau(p) \subseteq \tau(p_k) \in \lbrace U_1, \dots,
U_n \rbrace$. 
It remains to note
that 
for every formula $A' \in \Delta_i$ there exists
a trace through $\pi'$ that ends in $A'$, and therefore
also a possibly partial $\MC_\Gamma$-play according to $\exists$'s
winning strategy $g$ ending at $(A',x_{j+1})$.  This implies that for
all $A' \in \Delta_i$, $(A',x_{j+1})$ is a winning position for
$\exists$ in $\MC_\Gamma$, and hence $\model,x_{j+1} \models A'$ for all
$A'\in \Delta_i$.  
This finishes the proof of the claim and hence that of the theorem.
\end{proof}
\begin{exa}
Consider the following formula of the coalitional $\mu$-calculus
\[[C]\nu X.(p \land \overline{[N]}X) \land [D]\mu Y.(\bar p \lor [D]Y)\]
stating that ``coalition $C$ can achieve that, from the next stage
onwards, $p$ holds irrespective of the strategies used by other agents, and coalition $D$ can ensure (through suitable strategies used in the long term) that $\bar p$ holds after some finite number of steps''. Here, we assume that $C,D \subseteq N$ are such that $C \cap D = \emptyset$. Define a parity map $\Omega$ for the above formula by 
%\begin{sparitemize}
%\item 
$\Omega(\nu X.(p \land \overline{[N]}X)) = 2$, 
%\item 
$\Omega(\mu Y.(\bar p \lor [D]Y)) = 1$ and
%\item 
$\Omega(A) = 0$ otherwise.
The unsatisfiability of this formula is witnessed by the following closed tableau:
{\small
\[\xymatrix@R=-1.2ex{
   &  \\
  \underline{~[C] B \land [D] A~} %_{~[C] B \land [D] A~}
  \\
  \underline{~[C] B \,;\, [D] A~} %_{\frac{[C] p; [D] q}{p;q}} 
  &  \\
  \underline{\quad B \,;\, A \quad }  &  \\ %\ar `r [ru] `[uuu] `[uu]&  \\
	\underline{\quad p \land \overline{[N]} B\,;\, A \quad }  &  \\
	\underline{~ p \land \overline{[N]} B\,;\, \overline{p} \lor
	[D] A ~} &  \\
	\underline{\qquad \qquad p \,;\, \overline{[N]} B \,;\,
	\overline{p} \lor [D] A \qquad \qquad} &  \\
	          {\underline{{~p\,;\, \overline{[N]}B\,;\,\overline{p}~}} \qquad\qquad p\,;\,
	\overline{[N]} B\,;\, [D] A} 
	\ar `r `^l[ruuuu] [uuuu]  &   \\ % `[uuuu] &  \\
}
\]}
\noindent
where $B =  \nu X.(p \wedge \overline{[N]}X)$,
$A = \mu Y.(\bar p \vee [D]Y)$ and where we omitted the annotation
$\alpha$ because in this case $\alpha$ can be easily deduced from the structure of the tableau. 
For example, the annotation for the root node is equal to $[C] B \land [D] A$ and for the child of the root
the annotation is equal to $(\frac{[C] p; [D] q}{p;q},\sigma)$ where $\sigma: \{p,q\} \to \{A,B\}$
is a substitution with $\sigma(p) = B$ and $\sigma(q) = A$.
  
Any finite path through this
tableau ends in an axiom,
and the only infinite path contains the trace
\[[C]B \wedge [D]A,\, [D]A,\, A,\, \overline{ A,\, \bar p \vee
[D]A,\, \bar p \vee [D]A,\,[D]A,\, A } \]
where the overlined sequence is repeated ad infinitum.
This trace is bad with respect to
$\Omega$, as $\Omega(A) = 1$ and $A$ is the only fixpoint formula
that occurrs infinitely often.
\end{exa}

%------------------------------------------------------------------------- 

\section{The Tableau Game}\label{sec:tabgame}

We now introduce the tableau game associated to a clean and guarded sequent $\Gamma$, and use it to characterise the (non-)existence of
closed tableaux in terms of winning strategies in the tableau game.
For the entire section, we fix a modal similarity type $\Lambda$ and a
set $\Rules$ of monotone tableau rules that is both one-step sound
and complete.
The idea underlying the tableau game is that $\forall$
intends to construct 
a closed tableau for a given set of formulas $\Gamma$, 
while $\exists$ wants to demonstrate
that any tableau constructed by $\forall$ contains
a path $\pi$ that violates the closedness condition.    
As  infinite plays of the
tableau game correspond to paths through a tableau, an infinite play
should be won by $\exists$ if it does not carry a bad trace, that
is, outermost least fixpoints are only unfolded finitely often.
To be able to see this tableau game as a parity game, we therefore need a
mechanism to detect  bad traces, and we employ parity word automata
for this task. Board positions in the ensuing tableau game will
therefore be sequent / automata state pairs, with the priority of a
board position being determined by the parity function of the
automaton. In particular, this will ensure that winning strategies of
$\exists$ in the tableau game do not generate bad traces. 
We start our discussion of the tableau game by recalling some basic
notions concerning parity word automata.
\begin{defi}
Let $\Sigma$ be a finite alphabet.  A \emph{non-deterministic parity
$\Sigma$-word automaton} is a quadruple $\bbA=(Q,a_I,\delta:Q \times
\Sigma \to \Pow (Q),\Omega)$ where $Q$ is the set of states of
$\bbA$, $a_I \in Q$ is the initial state, $\delta$ is the transition
function,  and $\Omega: Q \to \omega$ is a (parity) function. Given
an infinite word $\gamma= c_0 c_1 c_2 c_3 \ldots$ over $\Sigma$, a
run of $\bbA$ on $\gamma$ is a sequence $\rho= a_0 a_1 a_2 \ldots
\in Q^\omega$ such that $a_0 = a_I$ and for all $i \in \omega$ we
have $a_{i+1} \in \delta(a_i,c_i)$.  A run $\rho$ is {\em accepting}
if $\rho$ is not a bad sequence with respect to $\Omega$.  We say
that $\bbA$ \emph{accepts} an infinite $\Sigma$-word $\gamma$ if
there exists an accepting run $\rho$ of $\bbA$ on $\gamma$.  Finally
we call $\bbA$ {\em deterministic} if $\delta(a,c)$ is a one-element
set
for all $(a,c) \in Q \times \Sigma$.  \end{defi}
\noindent
In other words, a parity word automaton is deterministic if its
transition function has type $Q \times \Sigma \to Q$. To develop the
tableau game, we use parity word automata over trace tiles
(cf.~Definition~\ref{defn:tiles}) to detect the existence of bad
traces through infinite plays. We now establish the existence of
such automata,  together with a
bound on both the state set and the range of the parity function.
\begin{lemmadef}\label{lem:parityword}
%\texttt{[adapt to the new notion of trace tile.]}
Let $\Gamma \in \Seq(\Lambda)$ be a clean, guarded sequent, and let
$\Sigma_\Gamma$ denote the set of trace tiles
$(\Delta,\flat,i)$ with $\Delta \in \Seq(\Gamma)$.
%and $i \leq
%\Cl(\Gamma)|$.
There exists a deterministic parity $\Sigma_\Gamma$-word automaton
$\bbA_\Gamma = (Q_\Gamma, a_\Gamma, \delta_\Gamma, \Omega')$ such
that $\bbA_\Gamma$ accepts an infinite sequence $(t_0, t_1, \dots) \in
\Sigma_\Gamma^\infty$ of trace tiles iff there is no sequence of
formulas $(A_0, A_1, \dots)$ with $(A_i, A_{i+1}) \in \Tr(t_i)$
which is a bad trace with respect to a parity function for $\Gamma$.
Moreover, the index of $\bbA$ and the cardinality of $Q$ are bounded
by $p(|\Cl(\Gamma)|)$ and $2^{p(|\Cl(\Gamma)|)}$ for a polynomial
$p$, respectively.  Such an automaton $\bbA$ is called a
\emph{$\Gamma$-parity automaton}.
\end{lemmadef}
\begin{proof}
We start by constructing a 
non-deterministic parity automaton that
accepts $w = t_0 t_1 t_2\ldots \in \Sigma_\Gamma^\omega$ iff $w$
does contain a sequence $A_0 A_1 \ldots \in \Cl(\Gamma)^\omega$ that
is bad w.r.t. $\Omega$ and satisfies $(A_i,A_{i+1}) \in
\Tr(t_i)$ for all $i \in \Nat$. We put $Q' =
\Cl(\Gamma)\cup\{a_I \}$ where we assume that $a_I \notin
\Cl(\Gamma)$ and define
$\delta': Q'\times \Sigma_\Gamma \to \Pow(Q')$ by 
$\delta'(a_I,t) = \bigcup_{A \in \Gamma} \Tr(t)(A) \subseteq
\Cl(\Gamma)$ and $\delta'(B,t) = \Tr(t)(B)$ for
$B \in \Cl(\Lambda)$ and $t \in \Sigma_\Gamma$.  If we put
$\Omega''(a_I) = 0$ and $\Omega''(B) = \Omega(B) + 1$ where $\Omega$
is a parity function for $\Gamma$, the automaton
$\bbA' = 
(Q',a_I,\delta',\Omega'')$ accepts a
word $w$ if $w$ \emph{does} contain a bad trace starting in some $B \in
\Gamma$.  We now transform $\bbA'$ into an equivalent deterministic parity
automaton $\bbA_d'$ by means of the Safra construction
to obtain an automaton of size $2^{O(nk\log(nk))}$ whose
parity function %$\Omega'''$ 
has a range of order $\mathcal{O}(nk)$ where
$n = |Q'| +1$ and $k$ is the cardinality of the range of $\Omega$
(cf.~\cite{pite:nond06,safr:onth88}).  The automaton $\bbA_\Gamma$
is then 
obtained by complementing $\bbA_d$ which can be done by changing
the parity function, and neither increases the size nor
the index of the automaton.  This implies the claim as the
cardinality $k$ of the range of $\Omega$ is bounded by the
size $n$ of the state set $n$ of the initial automaton.
\end{proof}

\noindent
We thus arrive at the following notion of tableau game, where
$\Gamma$-parity automata are used to detect bad traces.

\begin{defi}\label{defn:tableau-game}
Let $\Gamma \in \Seq(\Lambda)$ be clean and guarded, and let
$\bbA=(Q,a_\Gamma,\delta,\Omega)$ be a $\Gamma$-parity automaton.
We denote the set of tableau rules
$\Gamma_0/\Gamma_1,\ldots,\Gamma_n \in \Tab\Rules$ for which
$\Gamma_0 \in \Seq(\Gamma)$ by $\Tab\Rules_\Gamma$ and write
$\Bp(\Gamma)$ for the set of rule blueprints $\flat$ such that
\begin{enumerate}[$\bullet$]
\item $\flat \in \Cl(\Gamma)$ if $\flat \in \FoRm(\Lambda)$ and $A \in
\Cl(\Gamma)$ if $\flat = (A, \bar A)$
\item $\Gamma_0 \sigma \in \Seq(\Gamma)$ if $\flat = (r, \sigma)$ and
$r = \Gamma_0 / \Gamma_1 \dots \Gamma_n$.
\end{enumerate}
The \emph{$\Gamma$-tableau game} is the parity game
$\game_\Gamma = (B_\exists,B_\forall,E,\Omega')$  where
$B_\forall = \Seq(\Gamma) \times Q$, 
$B_\exists  = \Seq(\Gamma) \times \Bp(\Gamma) \times Q$  and
the relation $E \subseteq B_\forall \times B_\exists \cup B_\exists
\times B_\forall$ that defines the allowed moves is given by $(b_1, b_2)
\in E$ if either
\begin{enumerate}[$\bullet$]
\item $b_1 = (\Delta, a) \in B_\forall$, $b_2 = (\Delta, \flat, a)$
and $(\Delta, \flat)$ is a rule representation
\item  $b_1 =  (\Delta, \flat, a)$, $b_2 = (\Delta', a')$ and there
exists $i \in \Nat$ such that $\Delta'$ is the $i$-th conclusion of
the rule represented by $(\Delta, \flat)$ and $a' = \delta(a,
(\Delta, \flat, i))$.
\end{enumerate}
The parity function 
$\Omega': (B_{\exists} \cup B_\forall) \to \omega$ of $\game_\Gamma$
is given by $\Omega'(\Delta, a) = \Omega(a)$ if $(\Delta, a) \in
B_\forall$ and $\Omega'(\Delta,\flat, a) = 0$.
\end{defi}

  If not explicitly stated otherwise, we will only consider
  $\game_\Gamma$-plays that start at $(\Gamma,a_\Gamma)$
  where $a_\Gamma$ is the initial state of the automaton $\bbA$. 
  In particular,
  we say that a player has a winning strategy in $\game_\Gamma$
  if (s)he has a winning strategy in $\game_\Gamma$ at position
  $(\Gamma,a_\Gamma)$.

The easier part of the correspondence between satisfiability and
winning strategies in $\game_\Gamma$ is proved by
constructing a closed tableau based on a winning strategy for
$\forall$. To show that this tableau is indeed closed, we need to
show that every infinite path carries at least one bad trace, which
follows from the fact that $\forall$ wins in the tableau game. To
make this formal, we consider a notion of path and trace also
relative to plays in  the tableau game.

\begin{defi}\label{def:gametrace}
	For a $\game_\Gamma$-play 
	\[\pi=(\Gamma^0,a_0)
	(\Gamma^0,\flat_0,a_0)
	(\Gamma^1,a_1)
	(\Gamma^1,\flat_1,a_1)
	\ldots
	(\Gamma^l,a_l)
	(\Gamma^l,\flat_l,a_l)
	\ldots\]
	%with $ r_j = \Gamma^j/\Gamma^j_1,\cdots,\Gamma^j_{k(j)}
	%\quad \mbox{for all} \; j \in \Nat$, 
	a sequence 
          $\pi' = \Gamma^0 c_0 \Gamma^1
	c_1 \ldots \Gamma^l c_l \ldots$
	of sequents and natural numbers 
%	directional rule names is an
	is an
	{\em underlying path} of $\pi$
	if $t_i=(\Gamma^i,\flat_i,c_i)$ is a trace tile
%        $t_i = (\Gamma^i \setminus \Sigma_i,\flat_i,\Gamma^{i+1})$ is a 
%	consistent trace tile,
	and $\delta(a_i,t_i)=a_{i+1}$ for all $i \in \Nat$. 
	A sequence of formulas $\alpha= A_0 A_1 A_2 \ldots \in 
	\FoRm(\Lambda)^\infty$
	is a {\em trace} through $\pi$ if there exists
	an underlying path $\pi' = \Gamma^0 c_0\Gamma^1
	c_1 \Gamma^2 \ldots$ of $\pi$ such that
	$(A_i,A_{i+1}) \in \Tr(\Gamma^i,\flat_i,c_i)$ for all $i \in \Nat$.
\end{defi}
An underlying path of a $\game_\Gamma$-play is very similar to the notion of a tableau path.
This is due to the correspondence between tableaux and strategies of $\forall$
in the tableau game. This correspondence is crucial in the proof of the following theorem.

\begin{thm}\label{thm:abelardwinning}
Let $\Gamma \in \Seq(\Lambda)$ be clean and guarded. 
If $\forall$ has a winning strategy in 
$\game_\Gamma$, then  $\Gamma$ has a closed $\Tab\Rules$-tableau. 
%Here $a_\Gamma = \delta (a_I,R_\Gamma)$
%with $R_\Gamma = \{ (A,A) \mid A \in \Gamma \}$.
\end{thm}
\begin{proof}
  Suppose that $\forall$ has a winning strategy $f$ in $\game_\Gamma$
  at position $(\Gamma,a_\Gamma)$. As $\game_\Gamma$ is a parity game
  we can assume that $\forall$'s strategy is history-free, i.e.~it can be encoded
  as a partial function 
  $f: \Seq(\Gamma)\times Q \pto \Seq(\Gamma) \times \Bp(\Gamma) \times Q$.
  In order to prove the claim we are going to define a closed 
  tableau $\Tab=(N,K,R,\ell,\alpha)$ for $\Gamma$. 
    We define $N$ to be the set of positions in $\Seq(\Gamma) \times Q$ 
   for which $f$ is a winning strategy (in particular, this entails that
   $f$ is defined at all positions in $N$).
  % The set of nodes $N$ 
 % i$N = \Seq(\Gamma) \times Q$, 
  Obviously
%  , as $f$ is a winning strategy at position $(\Gamma,a_\Gamma)$,
  we have $(\Gamma,a_\Gamma) \in N$ and we put
  $R = (\Gamma,a_\Gamma)$.
%  and we 
%  the root $R$ is the pair $(\Gamma,a_\Gamma)$, 
  The labelling function on $N$ is the first projection 
  map, i.e.~$\ell(\Delta,a) = \Delta$
  for all $(\Delta,a) \in N \subseteq \Seq(\Gamma) \times Q$. 
  
  For all $(\Delta,a) \in N$ the set of $K$-successors 
  is defined using $\forall$'s
  strategy by putting $K(\Delta,a) = \{ 
  (\Delta',a') \mid (\Delta',a') \in E(f(\Delta,a)) \}$ 
  where $E(f(\Delta,a))$ is the set of possible moves
  of $\exists$ at $f(\Delta,a)$.
  Finally we define the annotation $\alpha$ of $\Tab$ by putting
  $\alpha(\Delta,a) = \pi_2 (f(\Delta,a))$ where $\pi_2: 
  \Seq(\Gamma) \times \Bp(\Gamma) \times Q \to \Bp(\Gamma)$ 
  denotes the second projection map.
  
  It is an easy consequence of the definition of the tableau game that $\Tab$ is a well-defined
  tableau. We now show that $\Tab$ is a closed tableau. To this aim consider first a finite
  complete
  path $\pi= (\Gamma_0,a_0)c_0(\Gamma_1,a_1) c_1\cdots 	
  c_{n-1}(\Gamma_n,a_n)$ through
  $\Tab$ with $(\Gamma_0,a_0) =  (\Gamma,a_\Gamma)$.  
   This gives rise to a $\game_\Gamma$-play of the form
  $$(\Gamma_0,a_0)(\Gamma_0,\flat_0,a_0)(\Gamma_1,a_1)(\Gamma_1,\flat_1,a_1)
  (\Gamma_2,a_2) \ldots (\Gamma_n,a_n)$$
  that is played according to $\forall$'s winning strategy $f$.
%  with an underlying path $\pi'=\Gamma_0 \flat_0 \Gamma_1 \flat_1\cdots 	
%  \flat_{n-1}\Gamma_n$.
  In order to see this, note that for all $0 \le i < n$
  we have $(\Gamma_{i+1},a_{i+1}) \in 
  E(f(\Gamma_i,a_i))$, i.e.~$(\Gamma_{i+1},a_{i+1})$ is a legal
  answer to $\forall$'s move at $(\Gamma_i,a_i)$
  if $\forall$ is playing according to his strategy $f$.  
    Since $\pi$ was assumed to be complete, and since $\forall$ has a winning strategy at the last position $(\Gamma_n,a_n)$ of the corresponding $\game_\Gamma$-play, it follows that $\exists$ cannot move in the position obtained by $\forall$ playing according to his strategy at $(\Gamma_n,a_n)$. This can only be the case if $\forall$ moves to $(\Gamma_n,(A,\bar A),a_n)$ at $(\Gamma_n,a_n)$ for some $A \in \FoRm(\Lambda)$, which in turn is only possible if $\Gamma_n$ is a tableau axiom.
 
  Consider now an infinite path $\pi = (\Gamma_0,a_0)
  c_0(\Gamma_1,a_1)c_1(\Gamma_2,a_2)\ldots$
  through $\Tab$ starting with $(\Gamma_0,a_0)=(\Gamma,a_\Gamma)$. 
  As in the previous case, 
  this induces an infinite $\game_\Gamma$-play $P$ of the form
  $$P = (\Gamma_0,a_0)(\Gamma_0,\flat_0,a_0)(\Gamma_1,a_1)(\Gamma_1,\flat_1,a_1)
  (\Gamma_2,a_2) \ldots (\Gamma_n,a_n)\ldots$$
  that is played according to $\forall$'s winning strategy $f$.
  By the definition of 
  the game board of $\game_\Gamma$, this means that
  the infinite sequence
  $\rho= 
  a_0 a_1 a_2\ldots \in Q^\omega$ can be seen
  as a run of $\bbA_\Gamma$ on 
  $$w =
  (\Gamma_0,\flat_0,c_0)
  (\Gamma_1,\flat_1,c_1)
  (\Gamma_2,\flat_2,c_2) \ldots
  \in \Sigma_\Gamma^\omega.$$
  By assumption 
  $f$ was winning for $\forall$
  and therefore $P$ does not satisfy the 
  parity condition $\Omega'$ of $\game_\Gamma$.
  This implies that $\rho=
  a_\Gamma a_1 a_2 \ldots \in Q^\omega$ 
  does not fulfil the parity condition $\Omega$
  of the automaton $\bbA_\Gamma$. In other words, as $\rho$
  is the run of $\bbA_\Gamma$ on $w$, there must be a 
  sequence $\beta = B_0 B_1 B_2 \ldots \in \Cl(\Gamma)^\omega$  
  such that $(B_i,B_{i+1}) \in \Tr(\Gamma_i,\flat_i,c_i)$ 
  that is bad w.r.t. $\Omega$. 
  In other words,
 $\beta$ is also a trace 
  through the path $\pi$, which implies 
  that there exists a trace through $\pi$ that is bad w.r.t. $\Omega$ 
  as required.  This finishes the proof 
  that $\Tab$ is closed.
\end{proof}
\noindent
The converse of the above theorem is established later as Theorem
\ref{thm:summary}. The challenge there is to construct a model for $\Gamma$ based on a winning strategy for $\exists$ in the $\Gamma$-tableau game. 
As we only allow substitution instances of modal (one-step) rules
that do not duplicate literals (we require that substitutions do not
decrease the cardinality of premises in one-step rules in Definition
\ref{defn:rules}), we need to require that 
the set of tableau rules to be closed under contraction.
\begin{defi}\label{def:contr_closed}
A set $\Rules$ of monotone one-step rules is \emph{closed under
contraction}, if for all rules $\Gamma_0 / \Gamma_1, \dots, \Gamma_n \in
\Rules$ and all $\sigma: \Var \to \Var$, there exists a rule
$\Delta_0 / \Delta_1, \dots, \Delta_k \in \Rules$ and a renaming
$\tau: V \to V$ such that $A \tau = B \tau$ for $A, B \in \Delta_0$
implies that $A = B$, $\Delta_0 \tau \subseteq \Gamma_0 \sigma$
and, for each $1 \leq i \leq n$, there exists $1 \leq j \leq k$ such
that $\Gamma_i \sigma \subseteq \Delta_j \tau$.
\end{defi}
\noindent
In other words, instances of one-step rules which duplicate
literals in the premise may be replaced by instances for
which this is not the case. 
\begin{rem}
      Every monotone $\Lambda$-structure admits a one-step tableau sound and 
      one-step tableau   complete set of monotone tableau rules that is closed under contraction.
      This follows from the fact that the set of one-step rules from 
      the proof of Proposition~\ref{propn:mon-existence}
      is closed under contraction:
      Consider a rule  $\Gamma_0 / \Gamma_1, \dots, \Gamma_n \in \Rules$ 
      and any renaming $\sigma$. Then, by the definition of the set of one-step rules $\Rules$ 
      in~Prop.~\ref{propn:mon-existence} we can easily show that
      $\Delta_0 / \Delta_1, \dots, \Delta_n \in \Rules$ with $\Delta_i=\Gamma_i \sigma$
      for $i=0,\dots,n$. Closure under contraction follows from the fact that  $\Delta_0 / \Delta_1, \dots, \Delta_n$ together 
      with
      $\tau=\mathrm{id}_\Var$ satisfy the conditions of Definition~\ref{def:contr_closed}.
\end{rem}
Under the condition of closure under contraction (cf.~Remark~\ref{rem:contra_closed}), we prove:

\begin{thm} \label{thm:eloise-wins}  Suppose that $\Gamma \in \Seq(\Lambda)$ is
clean and guarded and $\Rules$ is one-step tableau complete
and contraction closed. If
$\exists$ has a winning strategy in $\game_\Gamma$, then $\Gamma$
is satisfiable in a model of size $\Ord(2^{p(n)})$ where $n$ is the
cardinality of $\Cl(\Gamma)$ and $p$ is a polynomial.
\end{thm}
The proof of Theorem~\ref{thm:eloise-wins} constructs a model for $\Gamma$ out of the game board of $\game_\Gamma$ using a winning strategy $f$ for $\exists$ in $\game_\Gamma$. We use one-step tableau completeness to impose a $T$-coalgebra structure on those $\forall$-positions in $\game_\Gamma$ that are reachable through $f$-conform 
$\game_\Gamma$-plays, with the resulting coalgebra satisfying the truth lemma.
We then equip this $T$-coalgebra with a valuation that makes
$\Gamma$ satisfiable in the resulting model. While our construction
shares some similarities with the shallow model construction
of~\cite{Schroder:2008:PBR}, it is by no means a simple adaptation
of \emph{op.~cit.}, as we are dealing with fixpoint formulas and thus cannot employ induction over the modal rank of formulas to construct satisfying models. Our proof of satisfiability is also substantially different from the corresponding proof for the modal $\mu$-calculus (cf.~\cite{Niwinski:1996:GMC}) -- we show satisfiability by directly deriving a winning strategy for $\exists$ in the model-checking game from a winning strategy of $\exists$ in the tableau game.

We now turn to the details of the proof of Theorem~\ref{thm:eloise-wins}.
Throughout the proof, we assume that $\Gamma \in \Seq(\Lambda)$ is a
clean, guarded sequent and $f : \Seq(\Gamma) \times \Bp(\Gamma)
\times Q \pto \Seq(\Gamma) \times Q$ is a history-free winning strategy for $\exists$ in $\game_\Gamma$. The construction of a supporting Kripke frame for a model of $\Gamma$ is based on $\forall$-positions of $\game_\Gamma$ where only modal rules can be applied. This is formalised through the notion of \emph{atomic sequent}.
\begin{defi}
	A $\Lambda$-formula is {\em atomic} if it is either a propositional variable $p \in \Var$, a ne\-gated propositional
	variable $\bar p \in \bar\Var$, or a formula of the form $\hearts(A_1,\ldots,A_n)$ or $\bar \hearts(A_1, \dots, A_n)$. A sequent $\Delta \in \Seq(\Lambda)$ is {\em atomic} if all its elements are atomic. 
	We write $\atomic(\Gamma)$ for 
	the set of atomic sequents in $\Seq(\Gamma)$, and call a $\game_\Gamma$-position $(\Delta,a) \in B_\forall$ \emph{atomic} if $\Delta$ is atomic. 
\end{defi} 

The state set of the satisfying model that we are about to construct
are the atomic $\game_\Gamma$-positions $(\Delta,a)$ that are
reachable from $(\Gamma,a_\Gamma)$ through $\game_\Gamma$-play that
is played according to $f$. As the propositional rules are
invertible, we may assume that $\forall$ applies them in any fixed,
given order. This simplifies the model construction as it implies --
together with $\exists$'s strategy -- that every sequent is
unfolded to an atomic sequent in a unique way.  Fixing the order in
which $\forall$ applies propositional rules can be seen as a
strategy, that we call propositional:

\begin{defi}
     A {\em propositional strategy} for $\forall$ in the tableau game $\game_\Gamma$
     is a function 
     %$$\pstrat: \Seq(\Gamma)\setminus \atomic(\Gamma) \to \Tab\Rules_\Gamma$$
     $$\pstrat: \Seq(\Gamma)\setminus \atomic(\Gamma) \to \Bp(\Gamma)$$
     such that %$(\pstrat(\Delta),\emptyset,a)$ is a legitimate move at $(\Delta,a)$, 
$(\Delta,\pstrat(\Delta))$ is a rule representation 
for all $\Delta \in \Seq(\Gamma)\setminus \atomic(\Gamma)$. 
     A $\game_\Gamma$-play  is \emph{played according to $\pstrat$}
     if $\forall$ moves at any position of the form 
     $(\Delta,a) \in (\Seq(\Gamma)\setminus
     \atomic(\Gamma))\times Q$  that occurs in the play to the position 
%$(\pstrat(\Delta),\emptyset,a)$. 
$(\Delta,\pstrat(\Delta),a)$.      \end{defi}

For the remainder of this section we fix a propositional strategy $\pstrat$ for 
$\forall$. As annonced informally in the beginning, this dictates
that plays proceed to atomic positions in a unique way, and in fact
induces a function from arbitrary positions to atomic ones in the
tableau game.
\begin{lemmadef}
     Let $f$ be a strategy for $\exists$ in $\game_\Gamma$. 
%     and let
 %    $g$ be a $\pstrat$-conform strategy for $\forall$ in $\game_\Gamma$.
     For any position $(\Delta,a) \in \Seq(\Gamma) \times Q$ there exists precisely one
     position $(\Delta',a') \in \atomic(\Gamma) \times Q$ and one
     partial $\game_\Gamma$-play
     \[ (\Delta,a),\cdots,(\Delta',a') \]
     that is played according to $f$ and $\pstrat$ and which does not contain an instance of a modal rule.
     We let $\sigma_f : \Seq(\Gamma) \times Q \to \atomic(\Gamma)\times Q$ be the function given by
     $\sigma_f(\Delta,a) = (\Delta',a')$. 
\end{lemmadef}

For the construction of a satisfying model for $\Gamma$ we are going to define a relation
on the set of atomic positions of $\game_\Gamma$ where two
atomic positions are related if the second position is selected by
$\exists$'s strategy in response to $\forall$ playing a modal rule.
In the case of Kripke frames, this relation would already define the
satisfying model, but in the general case, we need to impose a
coalgebra structure on top of this relation in a coherent way. To
achieve this, we single out specific states (the $A$-successors) that we
take as under-approximation of the semantics of a formula $A$.
Informally speaking, an $A$-successor of an atomic state arises by
$\forall$ playing a modal rule, and $\exists$ selecting a conclusion
containing $A$ that is then reduced to another atomic position.
Formally, we introduce the notions of $A$-children (conclusions
selected by $\exists$ that contain $A$) and $A$-successors
(reductions of $A$-children to atomic form), both relative to a
strategy for $\exists$.
\begin{defi}
Suppose that $f$ is a history-free strategy of $\exists$ in
$\game_\Gamma$, and let $(\Delta, a) \in \atomic(\Gamma)$.
A position $(\Delta', a') \in \Seq(\Gamma) \times Q$ is an \emph{$A$-child} of $(\Delta, a)$
along $f$ if $A \in \Delta'$ and $(\Delta', a') = f(\Delta, 
\flat,a)$ where $((\Delta, a), (\Delta, \flat, a))$ is a legal move of
$\forall$ in $\game_\Gamma$. We put
\[ \Chld_f(A, \Delta, a) = \lbrace (\Delta', a')  \in \Seq(\Gamma) \times Q \mid (\Delta', a')
\mbox{ $A$-child of $(\Delta, a)$ along $f$} \rbrace \]
and write $\Chld_f(\Delta, a)$ for the collection of all $A$-children of
$(\Delta, a)$ along $f$.
An atomic position $(\Delta'', a'')$ is an \emph{$A$-successor} of
$(\Delta, a)$ along $f$ if $(\Delta'', a'') = \sigma_f(\Delta', a')$
for some $A$-child $(\Delta', a')$ of $(\Delta, a)$ along $f$. This
is denoted by
\[ \Succ_f(A, \Delta, a) = \lbrace (\Delta'', a'') \in \atomic(\Gamma) \times Q \mid \mbox{
$(\Delta'', a'')$ $A$-successor of $(\Delta, a)$ along $f$ } \rbrace
\]
and we write $\Succ_f(\Delta, a) = \bigcup_{A \in \Cl(\Gamma)} \Succ_f(A,\Delta,a)$ 
for the collection of
all $A$-successors of $(\Delta, a)$.
\end{defi}

In other words, an atomic position $(\Delta'', a'')$ is a successor of
$(\Delta, a)$ if it is reachable from $(\Delta, a)$ by a play that
is played according to $\exists$'s strategy $f$ and the (fixed)
propositional strategy $g$ that involves precisely one modal rule.
The position $(\Delta'', a'')$ is an $A$-successor of $(\Delta, a)$
if the conclusion of this modal rule that is picked by $f$ contains
the formula $A$. This allows us to introduce \emph{coherent}
coalgebra structures, i.e.~those structures on atomic positions that
satisfy the truth lemma.

\begin{defi}\label{def:premodel}
Suppose that $f$ is a history-free strategy for $\exists$ in
$\game_\Gamma$ and let
\[ Y = \lbrace (\Delta, a) \in \atomic(\Gamma)\times Q \mid \sigma_f(\Gamma, a_I)
\to^* (\Delta, a)  \rbrace \]
where for $(\Delta,a),(\Delta',a') \in \atomic(\Gamma) \times Q$, $(\Delta, a) \to (\Delta', a')$ if $(\Delta', a') \in
\Succ_f(\Delta, a)$. A coalgebra structure $\gamma: Y \to TY$ on $Y$ is
called \emph{coherent} if %the following two conditions hold:
%  \begin{enumerate}[(i)]
%     \item\label{premod1} for all $(\Delta,a) \in Y$ we have $\gamma(\Delta,a) \in
%         \F \suc_f(\Delta,a) $,
%      \item\label{premod2} 
\[\gamma(\Delta, a) \in \lsem \hearts \rsem_Y(\Succ_f(A_1, \Delta,
a),  \dots, \Succ_f(A_n, \Delta, a) )\]
whenever $\hearts(A_1, \dots, A_n) \in \Delta$.
%  \end{enumerate}
A valuation $h: \Var
\to \Pow(Y)$ is \emph{coherent} if $(\Delta, a) \in h(p)$ whenever
$p \in \Delta$.
\end{defi}

In other words, the carrier of a coherent coalgebra is the set of
atomic positions that are reachable from the initial position via
$\exists$'s strategy $f$, and the coalgebra structure is so that we
can establish the truth lemma, together with monotonicity of the
modal operators: the $A$-successors of an atomic position contain an
element of the disjunctive normal form of $A$ and hence serve as
an under-approximation of the truth-set of $A$. We note that a
position cannot be both an $A$-successor  and an $\bar A$-successor
of the same position.

\begin{lem}\label{lem:sanity}
	Let $f$ be a history-free winning strategy for $\exists$ in $\game_\Gamma$
	and let $(\Delta_1,a_1)$ and $(\Delta_2,a_2)$ be atomic $\game_\Gamma$-positions such that $f$ is a winning strategy for 
	$\exists$ at $(\Delta_1,a_1)$. Then for all formulas $A$ we have
	 \[  (\Delta_2, a_2) \in \Succ_f(A, \Delta_1, a_1) \quad
	 \mbox{implies that} \quad
	     (\Delta_2, a_2) \notin \Succ_f(\bar A, \Delta_1, a_1). \] 
\end{lem}
\begin{proof}
     Suppose for a contradiction that $(\Delta_2, a_2) \in \Succ_f(A, \Delta_1, a_1)$
     as well as $(\Delta_2, a_2) \in \Succ_f(\bar A, \Delta_1, a_1)$ for some formula $A$.
     Then, by the definition of $\Succ_f$, there must exist $(\Delta',a')$ and
     $(\Delta'',a'')$ in $\Seq(\Gamma) \times Q$ such that $A \in \Delta'$,
     $\bar A \in \Delta''$ and 
     $\sigma_f(\Delta',a') = \sigma_f(\Delta'',a'') = (\Delta_2,a_2)$.
     A straightforward induction argument shows that in this case there must exist
     a formula $B$ such that $B,\bar B \in \Delta_2$. Therefore $(\Delta_2,a_2)$
     is a winning position for $\forall$. But this contradicts the fact that there
     exists a $\game_\Gamma$-play from $(\Delta_1,a_1)$ to $(\Delta_2,a_2)$ played according to $f$, and
     our assumption that $f$ is winning at $(\Delta_1,a_1)$.
\end{proof}
\noindent
We now show that if $\exists$ has a winning strategy $f$ in the tableau game for $\Gamma$, 
then a coherent model for $\Gamma$ exists. This is where contraction closure is needed as the application of modal rules may not identify elements in the premise of a rule.

\begin{prop}\label{prop:premodexist}
Every history-free winning strategy $f : \Seq(\Gamma) \times \Bp(\Gamma) \times Q \pto \Seq(\Gamma) \times Q$ 
for $\exists$ in $\game_\Gamma$ induces a coherent model $(Y,
\gamma, h)$. 
\end{prop}
\begin{proof}
We follow Definition~\ref{def:premodel} and put $Y = \lbrace (\Delta, a) \in \atomic(\Gamma) \mid \sigma_f(\Gamma, a_I)
\to^* (\Delta, a)  \rbrace$ where $\to$ is as in the definition, and
we define a coherent valuation $h: \Var \to Y$ by $h(p) = \lbrace
(\Delta, a) \in Y \mid p \in \Delta \rbrace$. It remains to be seen
that we can define $\gamma: Y \to TY$ coherently.
It is a consequence of Lemma~\ref{lem:sanity}
  and of the fact that
  $f$ is a winning strategy for $\exists$ in $\game_\Gamma$ that  for all $(\Delta_1,a_1), (\Delta_2,a_2) \in Y$ we have
  \begin{equation}\label{equ:consistent}
      (\Delta_2, a_2) \in \Succ_f(A, \Delta_1, a_1) \quad
	 \mbox{implies} \quad
	     (\Delta_2, a_2) \notin \Succ_f(\bar A, \Delta_1, a_1).
  \end{equation}
   Now suppose for a contradiction that there is no $\gamma: Y \to \F Y$
   such that $(Y,\gamma)$ is a coherent coalgebra structure for $\Gamma$. Then there exists
   some $(\Delta,a) \in Y$ such that we cannot find a $t \in \F Y$ that
   satisfies the condition in Definition~\ref{def:premodel}.
   Consider the set of formulas
  \begin{equation}\label{magic_tableau}
     \begin{array}{l}
     \Theta =  \{\hearts (p_{A_1},\ldots,p_{A_n}) \mid \hearts
     (A_1, \ldots, A_n)  \in \Delta\}  \\ \hspace{2cm}\cup  \quad
     \{ \bar{\hearts} (p_{A_1},\ldots,p_{A_n}) \mid
     \bar{\hearts}( A_1,\ldots,A_n)  \in \Delta\}
     \end{array}
   \end{equation}
  where for any formulas of the form
   $\hearts(A_1,\ldots,A_n)$ or $\bar\hearts(A_1,\ldots,A_n)$ in $\Delta$
  we associate a unique propositional variable $p_{A_i}$ 
  to the formula $A_i$, for $i \in \{1,\ldots,n\}$. Let $V_\Theta$ be the set of propositional variables
  occurring in $\Theta$.
  We define a valuation $\tau: V_\Theta \to \Pow ( \Succ_f(\Delta,a))$ by putting $\tau(p_A) = \Succ_f(A,\Delta,a)$.
    %\[ \tau(p_A) = \{ (\Delta',a') \in \suc_f(\Delta,a) \mid (\Delta,a) \grel_f^A (\Delta',a') \} .\]

  Using our assumption on $(\Delta,a)$ it is not difficult to see that
  $\psem{\Theta}_{T \Succ_f(\Delta,a), \tau}=\emptyset$.
  Therefore one-step tableau completeness implies that there exists
  a rule $\Gamma_0 / \Gamma_1 \cdots \Gamma_n$ and a substitution
  $\sigma:V \to V$ such that $\Gamma_0 \sigma \subseteq \Theta$
   and $\psem{\Gamma_i \sigma}_{\Succ_f(\Delta,a),\tau}
   = \emptyset$ for all $i \in \{1,\ldots,n\}$. Because 
   of contraction closure of $\Rules$ we can assume w.l.o.g.\ that
   $\sharp(\Gamma_0 \sigma) = \sharp(\Gamma_0)$.
   
   On the other hand, for $\eta: V_\Theta \to \FoRm(\Lambda)$ with $\eta(p_A) = A$,
   we clearly have $\Gamma_0 \sigma \eta \subseteq \Delta$ with $\sharp(\Gamma_0 \sigma \eta)=\sharp(\Gamma_0)$, and
   thus $\forall$ can move in the tableau game from position $(\Delta,a)$
   to the position $(\Delta, (\Gamma_0/\Gamma_1 \cdots 
   \Gamma_n,\eta \circ \sigma),a)$.
   Now $\exists$ moves to some $(\Gamma_j \sigma \eta,a'')$ with
   $j \in \{1,\ldots,n\}$ according to her winning strategy $f$.
   Therefore we have $(\Gamma_j \sigma \eta,a'') \in \Chld_f(\Delta,a)$.
   Furthermore, the play can be continued according to $\exists$'s
   strategy $f$ until the atomic position $(\Delta',a') = \sigma_f(\Gamma_j \sigma \eta,a'')$
   is reached. By definition we have
   \begin{equation}\label{equ:contradict}
%   (\Delta,a) \grel_f^{B\eta}  (\Delta',a')
 (\Delta',a') \in \Succ_f(B\eta,\Delta,a) 
   \;\mbox{for all} \; B \in \Gamma_j \sigma . %\; \mbox{we have} \;
 %      \interval{(\Delta,a)}{(\Delta',a')} \simple B \eta.
   \end{equation}
%   We prove by induction on $B$ that
%   $(\Delta',a') \in Y$
%   and $(\Delta,a)\grel_f^{B\eta} (\Delta',a')$
%   implies $(\Delta',a')  \in \psem{B}_{\suc_f(\Delta,a),\tau}$:
   It now follows that $(\Delta',a')  \in \psem{B}_{\Succ_f(\Delta,a),\tau}$
   for all $B \in \Gamma_j \sigma$. To see this, consider an arbitrary formula $B \in \Gamma_j \sigma$.
   By the definition of $\Theta$ and the fact that $\Gamma_0 \sigma \subseteq \Theta$
   we have that $\Gamma_j \sigma$ consists of atoms only. 
   Therefore $B = p_A$ for some formula $A$. By (\ref{equ:contradict}), we know that
           $(\Delta',a') \in \Succ_f(p_A\eta,\Delta,a) = \Succ_f(A,\Delta,a)$, 
           %i.e.~$(\Delta,a)  \grel^A_f (\Delta',a')$
           and therefore
           $(\Delta',a')  \in \psem{B}_{\Succ_f(\Delta,a),\tau}$. 
   As $B$ was an arbitrary element of $\Gamma_j \sigma$ we obtain 
   $(\Delta',a') \in \psem{B}_{
   \Succ_f(\Delta,a),\tau}$
   for all $B \in \Gamma_j \sigma$, which
   contradicts the fact that
   $\psem{\Gamma_j \sigma}_{
   \Succ_f(\Delta,a),\tau} = \emptyset$. This concludes the proof.
\end{proof}
\noindent
\eat{
\begin{defi}
\label{defn:initialpos}
Let $f : \Seq(\Gamma) \times \Bp(\Gamma) \times Q \to \Seq(\Gamma) \times Q$ be a history-free winning strategy for $\exists$ in $\game_\Gamma$, and let $\bbY = (Y, \gamma, h)$ be the 
corresponding model of a coherent coalgebra structure $(Y,\gamma)$ for $\Gamma$ (cf.~Definition~\ref{defn:model}). We write $\MC_\Gamma$ for the model-checking game $\MC_\Gamma(\bbY)$, and call a $\MC_\Gamma$-position $(A,(\Delta,a))$ with $A \in \Gamma$ and $(\Delta,a) = \sigma_f(\Gamma,a_\Gamma)$ an \emph{initial} position.
\end{defi}
}
\noindent
We can now take a history-free winning strategy $f$ for $\exists$ in
the tableau game and show that the induced coherent model $Y$ satisfies
the initial sequent. This is achieved by converting the strategy $f$
(in the tableau game) to a strategy $\tilde f$ in the model checking
game over $Y$. Satisfiability then follows as soon as we establish
that $\MC_\Gamma$-plays that are played according to $\tilde f$ correspond to traces through $\game_\Gamma$-plays that are played according to $f$. 

\begin{lem}
\label{lem:mcstrategy}
Let $f$ be a history-free winning strategy for $\exists$ in $\game_\Gamma$,  let $\bbY = (Y, \gamma, h)$ be the 
coherent model induced by $f$, 
and consider a position $(A_0,(\Delta_0,a_0))$ in $\MC_\Gamma(\bbY)$ 
%be an initial position in $\MC_\Gamma(\bbY)$, that is,
with $(\Delta_0,
a_0) = \sigma_f(\Gamma, a_I)$ and $A_0 \in \Delta_0$.
Then $\exists$ has a strategy $\tilde f$ in $\MC_\Gamma(\bbY)$ at $(A_0,(\Delta_0,a_0))$ 
such that for any (possibly infinite) sequence
$(A_0,
(\Delta_0,a_0))(A_1,(\Delta_1,a_1))\ldots(A_n,(\Delta_n,a_n))\ldots$
that can be extended to an $\tilde f$-conform $\MC_\Gamma(\bbY)$-play by inserting positions of the form $(\hearts(B_1,
\dots, B_n), (U_1, \dots, U_n))$
we have
\begin{enumerate}[\em(1)]
\item there exists a (possibly infinite) $\game_\Gamma$-play $\pi$ and a trace $\tau = B_0, B_1, \ldots,B_r,\ldots$ through 
$\pi$ (cf. Def.~\ref{def:gametrace}), such that
  \begin{enumerate}[\em(a)]
\item $\pi$ contains a sub-sequence of $\forall$-positions of the form 
\[(\Delta'_0,a'_0),(\Delta'_1,a'_1),\ldots,(\Delta'_n,a'_n),\ldots\]
with $\sigma_f(\Delta'_i,a'_i) =
(\Delta_i,a_i)$ and $\Delta'_i \ni A_i$ for each $i \ge 0$
\item $\tau$ is contractable to $A_0,A_1,\ldots,A_n,\ldots$\;, that is, there exists an increasing sequence $0 = s_0 < s_1 < 
\dots$ of indices such that
$A_0 A_1 \ldots = B_{s_0} B_{s_1} \ldots$, where $B_i = B_{s_j}$ whenever $s_j \leq i  < s_{j+1}$, for $j=0,1,\ldots$.
  \end{enumerate}
\item for all $\MC_\Gamma(\bbY)$-positions of the form $(A,(\Delta,a))$ occurring in $\pi$, with $A$ atomic, we have $A \in \Delta$.
\end{enumerate}
\end{lem}
\begin{proof}
    We define the strategy $\tilde f$ for $\exists$ in 
		$\MC_\Gamma(\bbY)$ starting at position $(A_0,(\Delta_0,a_0))$ by showing
		how to extend each partial, $\tilde f$-conform $\MC_\Gamma(\bbY)$-play starting in $(A_0,(\Delta_0,a_0))$ and ending in an $\exists$-position $b = (B,(\Delta,a))$ 
    with a position $b'$, such that $(b,b')$ is a valid move for
		$\exists$ in $\MC_\Gamma(\bbY)$. 
    We will show later that each such partial play determines a partial $\game_\Gamma$-play starting in $(\Delta_0,a_0)$ and ending in some $(\Delta',a') \in \Seq(\Gamma) \times Q$ with $\sigma_f(\Delta',a')=(\Delta,a)$ 
    and $\Delta' \ni B$.
%     (and as a result, 
 %   $(\Delta,a) \vdash_{\mathsf{SL}} \phi$). 
    At this point, we assume the above, and base our definition of $\exists$'s strategy solely on 
    $(B,(\Delta,a))$ and $(\Delta',a')$. We define $\exists$'s move in $(B,(\Delta,a))$ by case analysis on $B$:
\begin{enumerate}[\hbox to8 pt{\hfill}]
\item\noindent{\hskip-12 pt\bf Case $B = B_1 \vee B_2$:}\ Then $\sigma_f(\Delta',a')=(\Delta,a)$ together with $\Delta' \ni B$ ensure the existence of a $\game_\Gamma$-play of the form 
$$(\Gamma_0,d_0)(\Gamma_0,\flat_0,d_0) \ldots (\Gamma_{k-1},\flat_{k-1},d_{k-1})(\Gamma_k,d_k)$$ 
with $(\Gamma_0,d_0)=(\Delta',a')$, $\Gamma_j \not\in \atomic(\Gamma)$ for 
$0 \le j < k$ and $(\Gamma_k,d_k)=(\Delta,a)$,  
%$(r_0,j(0)), \ldots,(r_k,j(k))$, 
that is played according to $f$ and $\pstrat$, such that %$b_{2i+1} = (\Delta'';B_1 \vee B_2/\Delta'';B_1,\Delta'';B_2,\emptyset,a'')$
$\flat_j = B_1 \vee B_2$ for some $0 \le j < k$. 
Let %$\Delta_0 \flat_0 \ldots \flat_{k-1} \Delta_k$
$\Gamma_0 c_0 \ldots c_{k-1} \Gamma_k$
 be an underlying path of the above $\game_\Gamma$-play.
Then, %$\flat_{i}$ is either equal to $\vee_l$ or to $\vee_r$.
$c_{j} \in \{1,2\}$, and 
%
%$B$ appears as a formula in $b_i$, and $B_{j(i)}$ appears as a formula in $b_{i+1}$, 
%for some $0 \le i < k$. 
we define $\exists$'s move at position 
$(B,(\Delta,a))$ of $\MC_\Gamma(\bbY)$ to be to the position %$(B_{j(i)},(\Delta,a))$ , where $j(i) = 1$ if $\flat_{i}=\vee_l$, and $j(i) = 2$ if $\flat_{i}=\vee_r$.
%if there exists such a path with $j(i)=1$, 
$(B_{c_{j}},(\Delta,a))$. 
Moreover, we note for future reference that the tableau 
node $(\Gamma_{j+1},d_{j+1})$ satisfies $\sigma_f(\Gamma_{j+1},d_{j+1})= (\Delta,a)$ and $\Gamma_{j+1} \ni B_{c_j}$.
\item\noindent{\hskip-12 pt\bf Case $B = \hearts (B_1,\ldots,B_n)$:}\ We define
$\exists$'s move at position $(B,(\Delta,a))$ of $\MC_\Gamma(\bbY)$ to be 
to the position 
$(B,(U_1,\ldots,U_n))$
with 
$$
\begin{array}{lcl}
U_j & = &\Succ_f(B_j,\Delta,a)
%U_j & = &\{(\Delta'',a'') \mid (\Delta,a) \grel_f^{B_j} (\Delta'',a'') \} %~\text{and}~ \interval{(\Delta,a)}{(\Delta'',a'')\}
\end{array}$$
for $j = 1,\ldots,n$. To justify this move, we must show that $\gamma(\Delta,a) \in \lsem \hearts \rsem_Y(U_1,\ldots,U_n)$.
But this follows from Definition~\ref{def:premodel}.
\end{enumerate}
\noindent
This defines a strategy for $\exists$ as there is no choice for
$\exists$ at all other positions $(B,(\Delta,a))$
in $\MC_\Gamma(\bbY)$.
Now consider a (possibly infinite) $\MC_\Gamma(\bbY)$-play of the form
\[(A_0, (\Delta_0,a_0)),(A_1,(\Delta_1,a_1)),\ldots,(A_n,(\Delta_n,a_n)),\ldots\]
played according to the previously defined strategy. 
% with the additional property that $(\Delta,a) \in Y$ and $(\Delta,a) \vdash_{\mathsf{SL}} A$. 
We shall construct a $\game_\Gamma$-play $\pi$ and an underlying path $\pi'$ of $\pi$ with an associated trace $\tau$, with the required properties. In particular, the construction of $\pi$ will supply a sequence of $\game_\Gamma$-positions $(\Delta_0',a_0'), (\Delta_1',a_1'), \ldots$ to be used in defining $\exists$'s moves.

To begin with, note that by assumption on $(A_0, (\Delta_0,a_0))$ 
%being an initial $\MC_\Gamma(\bbY)$-position,
we have $\sigma_f(\Gamma,a_\Gamma) = (\Delta_0,a_0)$ and $A_0 \in \Gamma$. Hence, we let $(\Delta_0',a_0') = (\Gamma,a_\Gamma)$ be the first position
of $\pi$, let $\Gamma$ be the first position of $\pi'$, and let $\tau_0 = A_0 \in \Gamma$.

Now assume that $\pi$, $\pi'$ and $\tau$ have been constructed up to a position $(\Delta'_i,a'_i)$, with $\sigma_f(\Delta'_i,a'_i)=(\Delta_i,a_i)$ and $\Delta'_i \ni A_i$. We extend the partial $\game_\Gamma$-play $\pi$ with a segment starting in $(\Delta'_i,a'_i)$ and ending in some $(\Delta'_{i+1},a'_{i+1})$, with $\sigma_f(\Delta'_{i+1},a'_{i+1})=(\Delta_{i+1},a_{i+1})$ and $\Delta'_{i+1} \ni A_{i+1}$. Here $(\Delta_{i+1},a_{i+1})$ represents the position obtained as a result of $\exists$ moving in $(\Delta_i,a_i)$, based on the additional information provided by $(\Delta_i',a_i')$, according to the strategy defined earlier. At the same time, we extend the underlying path $\pi'$ of $\pi$ with a segment $\Delta_i' \ldots \Delta_{i+1}'$, and the trace $\tau$ with a segment $A_i, \ldots, A_i, A_{i+1}$. These constructions are carried out by case analysis on $A_i$.
 \begin{enumerate}[\hbox to8 pt{\hfill}]
\item\noindent{\hskip-12 pt\bf Case $A_i = A_i^1 \vee A_i^2$:}\ Here, the
$\MC_\Gamma(\bbY)$-move from $(A_i,(\Delta_i,a_i))$ to $(A_{i+1},(\Delta_{i+1},a_{i+1}))$ is an $\exists$-move played according to the strategy defined earlier. The definition of this move was based on a $\game_\Gamma$-play of the form
$$(\Gamma_0,d_0)(\Gamma_0,\flat_0,d_0) \ldots(\Gamma_{k-1},\flat_{k-1},d_{k-1})(\Gamma_k,d_k)$$
with $(\Gamma_0,d_0)=(\Delta_i',a_i')$, $\Gamma_l \not\in \atomic(\Gamma)$ for $0 \le l < k$ and $(\Gamma_k,d_k)=(\Delta_i,a_i)$, played according to $f$ and $\pstrat$, with an underlying path $\Gamma_0 c_0 \ldots c_{k-1} \Gamma_k$, such that there exists $0 \le j < k$ with $(A_i,A_i) \in \Tr(\Gamma_l,\flat_l,c_l)$ for $0 \le l < j$ and $(A_i,A_i^{c_{j}}) \in \Tr(\Gamma_{j},\flat_{j},c_{j})$. 
Moreover, this definition guarantees that we have
$\sigma_f(\Gamma_{j+1},d_{j+1}) = (\Delta_{i+1},a_{i+1}) =
(\Delta_i,a_i)$. We now put
$(\Delta_{i+1}',a_{i+1}')=(\Gamma_{j+1},d_{j+1})$, and extend the
play $\pi$ to $(\Gamma_0,\flat_0,d_0) \ldots
(\Gamma_{j},\flat_{j},d_{j}) (\Gamma_{j+1},d_{j+1})$, the underlying
path $\pi'$ with $c_0 \ldots c_j \Gamma_{j+1}$, and the trace $\tau$
with $A_i, \ldots,A_i,A_i^{c_{j}}$.

\item\noindent{\hskip-12 pt\bf Case $A_i = A_i^1 \wedge A_i^2$:}\ This time, the move from $(A_i,(\Delta_i,a_i))$ to $(A_{i+1},(\Delta_{i+1},a_{i+1}))$ is a $\forall$-move, with $A_{i+1} = A_i^l$ for some $l \in \{1,2\}$ and $(\Delta_{i+1},a_{i+1}) = (\Delta_i,a_i)$. Since $A_i \in \Delta_i'$ and $\sigma_f(\Delta_i',a_i') = (\Delta_i,a_i)$, it follows that there exist a $\game_\Gamma$-play of the form 
$$(\Gamma_0,d_0)(\Gamma_0,\flat_0,d_0) \ldots (\Gamma_{k-1},\flat_{k-1},d_{k-1})(\Gamma_k,d_k)$$
with $(\Gamma_0,d_0)=(\Delta_i',a_i')$, $\Gamma_l \not\in \atomic(\Gamma)$ for $0 \le l < k$ and $(\Gamma_k,d_k)=(\Delta_i,a_i)$, played according to $f$ and $\pstrat$, such that $\flat_j = A^1_i \wedge A^2_i$ for some $0 \le j < k$, and an underlying path $\Gamma_0 c_0 \ldots
c_{k-1} \Gamma_k$ of this $\game_\Gamma$-play that satisfies $(A_i,A_i) \in \Tr(\Gamma_h,\flat_h,c_h)$ for $0 \le h < j$ and $(A_i,A_i^{l}) \in \Tr(\Gamma_j,A^1_i \wedge A^2_i,c_j)$. From the latter we obtain $l = c_j$. We then let $(\Delta_{i+1}',a_{i+1}')$ be given by 
$(\Gamma_{j+1},d_{j+1})$, and note that
$\sigma_f(\Delta_{i+1}',a_{i+1}')=(\Delta_{i+1},a_{i+1}) =
(\Delta_i,a_i)$ and $\Delta_{i+1}' = \Gamma_{j+1} \ni A_i^{c_j} =
A_i^l = A_{i+1}$. It is therefore possible for us to extend the play $\pi$
with the sequence $(\Gamma_0,\flat_0,d_0) \ldots
(\Gamma_{j},\flat_{j},d_{j})(\Gamma_{j+1},d_{j+1})$, the underlying
path $\pi'$ with $c_0 \ldots c_{j} \Gamma_{j+1}$, 
and the trace $\tau$ with $A_i, \ldots, A_i,A_i^l$.

\item\noindent{\hskip-12 pt\bf Case $A_i = \hearts (B_1,\ldots,B_n)$:}\ The move from $(A_i,(\Delta_i,a_i))$ to $(A_{i+1},(\Delta_{i+1},a_{i+1}))$ thus incorporates an $\exists$-move played according to the strategy defined earlier, followed by a $\forall$-move. Again, from $A_i \in \Delta_i'$ and $\sigma_f(\Delta_i',a_i')=(\Delta_i,a_i)$ we obtain a $\game_\Gamma$-play of the form 
$$(\Gamma_0,d_0)(\Gamma_0,\flat_0,d_0) \ldots (\Gamma_{k-1},\flat_{k-1},d_{k-1})(\Gamma_k,d_k)$$
with $(\Gamma_0,d_0)=(\Delta_i',a_i')$, $\Gamma_l \not\in
\atomic(\Gamma)$ for $0 \le l < k$ and
$(\Gamma_k,d_k)=(\Delta_i,a_i)$, played according to $f$ and
$\pstrat$, that has an underlying path $\Gamma_0 c_0 \ldots c_{k-1}
\Gamma_k$ such that $(A_i,A_i) \in \Tr(\Gamma_j,\flat_j,c_j)$ for $0
\le j < k$. Also, by definition of $\exists$'s move in  $(A_i,(\Delta_i,a_i))$
we obtain $A_{i+1} = B_j$ and $(\Delta_{i+1},a_{i+1}) \in \Succ_f(B_j,\Delta_i,a_i)$ %$(\Delta_i,a_i) \grel_f^{B_j} (\Delta_{i+1},a_{i+1})$
%~\text{and}~ \interval{(\Delta_i,a_i)}{(\Delta',a')} \simple B_j
for some $j \in \{1,\ldots,n\}$. 
It follows that there exists a position $(\Delta'',a'')$ such that $(\Delta'',a'') \in \Chld_f(\Delta_i,a_i)$, $B_j \in \Delta''$ and $\sigma_f(\Delta'',a'') = (\Delta_{i+1},a_{i+1})$.
%It now follows that there exists an $f$-conform $\game_\Gamma$-play $b_k \ldots b_l$ that has an underlying play path $\Gamma_k \flat_k \ldots \flat_{l-1} \Gamma_l$, with $(\hearts(B_1,\ldots,B_n),B_j) \in \Tr(\Gamma_k,\flat_k,\Gamma_{k+1})$ and $b_l = (\Delta_{i+1},a_{i+1})$. 
We then let $(\Delta'_{i+1},a'_{i+1})$ be given by $(\Delta'',a'')$.
%, and note that $(\Delta'_{i+1},a'_{i+1}) \to_f (\Delta_{i+1},a_{i+1})$ and $\Delta'_{i+1} \ni B_j$ hold. 
Moreover, from $(\Delta'_{i+1},a'_{i+1}) \in \Chld_f(\Delta_i,a_i)$
it follows that $\forall$ can move in $\game_\Gamma$ from
$(\Delta_{i},a_{i})$ to some $(\Delta_{i},\flat,a_{i})$ with
$f(\Delta_{i},\flat,a_{i}) = (\Delta'_{i+1},a'_{i+1})$. Since
$\exists$'s move  at position $(\Delta_{i},\flat,a_{i})$ was legal,
this now yields $c \in \Nat$ such that $\Delta'_{i+1}$ is the $c$-th
conclusion of the rule represented by $(\Delta_{i},\flat)$. This
together with $B_j \in \Delta'_{i+1}$ yield $(A_i,B_j) \in
\Tr(\Delta_{i},\flat,c)$. It is now possible to extend the play
$\pi$ with 
\[(\Gamma_0,\flat_0,d_0) \ldots
(\Gamma_{k-1},\flat_{k-1},d_{k-1})(\Gamma_k,d_k) (\Delta_i,\flat,a_i)
(\Delta'_{i+1},a'_{i+1})\ ,
\] 
the underlying path $\pi'$ with $c_0 \ldots c_{k-1} \Gamma_k c \Delta_{i+1}'$, 
and the trace $\tau$ with $A_i, \ldots, A_i , B_j$.

\item\noindent{\hskip-12 pt\bf Case $A_i = \eta X.A$, $\eta \in
\{\mu,\nu\}$:}\ The move from $(A_i,(\Delta_i,a_i))$ to
$(A_{i+1},(\Delta_{i+1},a_{i+1}))$ consists of unfolding the fixpoint
variable $X$, that is, $A_{i+1} = A[X := \eta X.A]$ and
$(\Delta_{i+1},a_{i+1}) = (\Delta_i,a_i)$. Again, $\Delta'_i \ni \eta
X.A$ together with $\sigma_f(\Delta'_i,a'_i) =(\Delta_i,a_i)$ yield a
$\game_\Gamma$-play of the form
$$(\Gamma_0,d_0)(\Gamma_0,\flat_0,d_0) \ldots (\Gamma_{k-1},\flat_{k-1},d_{k-1})(\Gamma_k,d_k)$$
with $(\Gamma_0,d_0)=(\Delta'_i,a'_i)$, $\Gamma_l \not\in \atomic(\Gamma)$ for $0 \le l < k$ and $(\Gamma_k,d_k)=(\Delta_i,a_i)$, played according to $f$ and $\pstrat$, such that $\flat_j=\eta X.A$ for some $0 \le j < k$, and an underlying path $\Gamma_0 c_0 \ldots c_{k-1} \Gamma_k$ of this $\game_\Gamma$-play that satisfies $(A_i,A_i) \in \Tr(\Gamma_h,\flat_h,c_h)$ for $0 \le h < j$ and $(A_i,A_{i+1}) \in \Tr(\Gamma_{j},\flat_j,c_j)$. We now let $(\Delta'_{i+1},a'_{i+1})$ be given by $(\Gamma_{j+1},d_{j+1})$, 
and note that
$\sigma_f(\Delta_{i+1}',a_{i+1}')=(\Delta_{i+1},a_{i+1}) =
(\Delta_i,a_i)$ and $\Delta_{i+1}' \ni A[X := \eta X.A]$. It is
therefore possible to extend the play $\pi$ with
$(\Gamma_0,\flat_0,d_0) \ldots
(\Gamma_{j},\flat_{j},d_{j})(\Gamma_{j+1},d_{j+1})$, the underlying
path $\pi'$ with $c_0 \ldots c_j \Gamma_{j+1}$, and the trace $\tau$
with $A_i, \ldots, A_i,A[X := \eta X.\phi]$.
  \end{enumerate}

\noindent To show the second property of the $\MC_\Gamma(\bbY)$-play 
\[(A_0, (\Delta_0,a_0)),(A_1,(\Delta_1,a_1)),\ldots,(A_n,(\Delta_n,a_n)),\ldots\]
we note that $\sigma_f(\Delta_i',a_i')=(\Delta_i,a_i)$ together with $A_i$ atomic and $A_i \in \Delta_i'$ yield $A_i \in \Delta_i$, for $i = 0,1,\ldots$.
\end{proof}

\noindent
Finally, we prove satisfiability of $\Gamma$ in $\bbY$ by showing
that the strategy resulting from Lemma~\ref{lem:mcstrategy} is a
winning strategy for $\exists$ in $\MC_\Gamma(\bbY)$.

\begin{thm}
\label{thm:existssatisfiability}
Let $f : \Seq(\Gamma) \times \Bp(\Gamma) \times Q \pto \Seq(\Gamma) \times Q$ be a history-free winning strategy for $\exists$ in $\game_\Gamma$, and let $\bbY = (Y, \gamma, h)$ be the 
corresponding model of a coherent coalgebra structure $(Y,\gamma)$ for $\Gamma$. Then, $\bbY,(\Delta,a) \models A$ for all states $(\Delta,a) \in \sigma_f(\Gamma,a_\Gamma)$ and all formulas $A \in \Gamma$.
\end{thm}
\begin{proof}
Let $(\Delta_0,a_0) \in Y$ be such that
$\sigma_f(\Gamma,a_\Gamma)=(\Delta_0,a_0)$, and let $A_0 \in
\Gamma$. Thus, $(A_0,(\Delta_0,a_0))$ is an initial position of
$\MC_\Gamma(\bbY)$. Let $\tilde f$ be the strategy for $\exists$ at
$(A_0,(\Delta_0,a_0))$ in $\MC_\Gamma(\bbY)$ provided by
Lemma~\ref{lem:mcstrategy}. We show that $\bbY,(\Delta_0,a_0)
\models A_0$ by showing that $\exists$ wins all $\MC_\Gamma(\bbY)$-plays that start at position $(A_0,(\Delta_0,a_0))$ and are played according to $\tilde f$.

Consider such a play, and assume first that it is finite. Let $(A,(\Delta,a))$ be its last position of type $\Cl(\Gamma) \times (\Seq(\Gamma) \times Q)$. Thus, the last position of the play is either $(A,(\Delta,a))$ itself, or a $\forall$-position of type $(\hearts(B_1,\ldots,B_n),(U_1,\ldots,U_n))$, with $U_i = \emptyset$ for $i = 1,\ldots,n$. In either case, $A$ is atomic (otherwise the play would not be complete). We distinguish the following cases:

 \begin{enumerate}[(1)]
 \item $A = p$ for some propositional variable $p$. By coherence of
 the valuation, 
we have $p \in \Delta$, and therefore by the definition of $\bbY$ we have $(\Delta,a) \in h(p)$, which
implies that $(p,(\Delta,a))$ is a winning position for $\exists$.
\item $A = \bar p$. Similar to the previous case.
\item $A = \hearts(B_1,\ldots,B_n)$. According to the definition of $\exists$'s strategy $\tilde f$, the last position of the play must be a $\forall$-position of type $(\hearts(B_1,\ldots,B_n),(U_1,\ldots,U_n))$ with $U_i = \emptyset$ for $i = 1,\ldots,n$ (as $\exists$ can always play in positions of type $(\hearts(B_1,\ldots,B_n),(\Delta,a))$). Thus, $(A,(\Delta,a))$ is a winning position for $\exists$.
 \end{enumerate}
It therefore follows that $\exists$ wins all finite
$\MC_\Gamma(\bbY)$-plays that start at $(A_0,(\Delta_0,a_0))$ and
are played according to $\tilde f$. Now consider an infinite
$\MC_\Gamma(\bbY)$-play starting at $(A_0,(\Delta_0,a_0))$ and
played according to $\tilde f$, and let $\pi$ be the infinite
$\game_\Gamma$-play and $\tau$ be the associated trace through $\pi$
provided by Lemma~\ref{lem:mcstrategy}. It follows from the
statement of the lemma that $\tau$ is contractable to the sequence
of formulas appearing in the given $\MC_\Gamma(\bbY)$-play. Since
the strategy $f$ was winning for $\exists$ in $\game_\Gamma$, it
follows that any trace through $\pi$, and therefore also $\tau$,
satisfies the parity condition of $\game_\Gamma$. As a result, the
parity condition of $\MC_\Gamma(\bbY)$ is satisfied by the given
infinite $\MC_\Gamma(\bbY)$-play, which is thus won by $\exists$.
\end{proof}

Theorem~\ref{thm:eloise-wins} now follows from Theorem~\ref{thm:existssatisfiability} and the observation that the sizes of both $Q$ and $\Seq(\Gamma)$ are bounded by an exponential in the size of $\Cl(\Gamma)$ (by Lemma~\ref{lem:parityword} and respectively the definition of $\Seq(\Gamma)$).

Putting everything together, we obtain a complete characterisation of satisfiability in the coalgebraic $\mu$-calculus.
\begin{thm} \label{thm:summary}
Suppose that $\Gamma \in \Seq(\Lambda)$ is a clean, guarded sequent and $\Rules$ is one-step tableau complete and
contraction closed. Then %the following are equivalent:
%\begin{sparenumerate}
%\item 
$\Gamma$ is satisfiable
iff %\item 
no tableau for $\Gamma$ is closed
iff %\item 
$\exists$ has a winning strategy in the tableau game
$\game_\Gamma$.
%\end{sparenumerate}
\end{thm}
\noindent
As a by-product, we obtain the following small model property.
\begin{cor}
A satisfiable, clean and guarded formula $A$ is satisfiable in a
model of size $\Ord(2^{p(n)})$ where $n$ is the cardinality of
$\Cl(A)$ and $p$ is a polynomial.
\end{cor}
\begin{proof}
The statement follows immediately from Theorems~\ref{thm:noclosedtab}, \ref{thm:abelardwinning} and \ref{thm:eloise-wins} together with the determinacy of two player parity games.
\end{proof} 
%------------------------------------------------------------------------- 
%------------------------------------------------------------------------- 
%------------------------------------------------------------------------- 
\section{Complexity}

We now show that -- subject to a mild condition on the rule set --
the satisfiability problem for guarded formulas of the coalgebraic $\mu$-calculus is
decidable in exponential time. By Theorem \ref{thm:summary}, the
satisfiability problem is reducible to the existence of
winning strategies in parity games. Given any guarded sequent $\Gamma$, we thus
construct a parity game of exponential size (measured in the size of
$\Gamma$), the parity function of which has polynomial range (again
measured relative to the size of $\Gamma$). 
This will ensure
\textsc{Exptime}-decidability if we can decide legal moves in this
game in exponential time. 
According to Definition \ref{defn:tableau-game}, the game board
consists of the disjoint union of
\begin{enumerate}[$\bullet$]
\item $\Seq(\Gamma) \times Q$ (the positions owned by $\forall$)
where $Q$ is the state set of a $\Gamma$-parity automaton and
$\Seq(\Gamma)$ are the sequents that we can form in the closure of
$\Gamma$, and
\item $\Seq(\Gamma) \times \Bp(\Gamma) \times Q$ where $\Bp(\Gamma)$
are the blueprints of rules with premise in $\Seq(\Gamma)$.
\end{enumerate}\medskip

\noindent We know that the state set $Q$ of the $\Gamma$-parity automaton is
exponential in the size of $\Cl(\Gamma)$ by Lemma
\ref{lem:parityword} and it is easy to see that $\Seq(\Gamma)$ is
exponentially bounded. The crucial step for obtaining an overall
exponential bound is thus the ability to treat rule blueprints.
While this is simple for many logics (where it is easy to see one
only has exponentially many applicable rule/substitution pairs that
are of polynomial size), more care is needed for the rules of the
probabilistic and the graded $\mu$-calculus. The main difficulty
lies in the fact that the conclusions of these rules (Example
\ref{example:rules}) are sets of sequents that may be exponentially
large. On the other hand, the conclusions can be represented by
(small) linear
inequalities, as in fact we did in Example \ref{example:rules} for
presentational purposes, and leads to an obvious solution. Instead
of representing rule blueprints associated with modal rules
directly, we use a coding of modal rules that can be decided
efficiently, to obtain an exponential representation of the game
board. This approach has been used previously in
\cite{Schroder:2008:PBR} to give \textsc{Pspace}-bounds for
coalgebraic logics, and we will refer to \emph{op.cit.} for some of
the technical points.

In order to be able to speak about the complexity of the satisfiability
problem in a meaningful way, we begin by formalising the notion of 
size of formulas and sequents.
To do this, we assume that the underlying
similarity type $\Lambda$ is equipped with a size measure $s:
\Lambda \to \Nat$ and measure the size of a formula $A$ in terms
of the number of subformulas counted with multiplicities, adding $s(\hearts)$ for every
occurrence of a modal operator $\hearts \in \Lambda$ in $A$.  In the
examples, we code numbers in binary, that is, $s(\langle k \rangle)
= s([ k ] ) = \lceil \log_2 k \rceil$ for the graded $\mu$-calculus
and $s(\langle p / q \rangle) = s([ p / q ]) = \lceil
\log_2 p \rceil + \lceil \log_2 q \rceil + 1$ for the probabilistic
$\mu$-calculus, and $s([a_1, \dots, a_k]) = 1$ for coalition logic.
Note that in the latter case, the overall number of agents is fixed,
so there will only be finitely many coalitions which allows us to
assign unit size to every operator.
The definition of size is extended to sequents by 
$\size(\Gamma) = \sum_{A \in \Gamma} \size(A)$ for $\Gamma \in
\Seq(\Lambda)$ and $\size(\lbrace \Gamma_1, \dots, \Gamma_n \rbrace)
= \sum_{i = 1}^n \size(\Gamma_i)$ for sets of sequents.

We continue by discussing the mechanism to encode rule blueprints
that we did describe informally at the beginning of this section. In
order to obtain an exponential bound, we require that blueprints of
modal rules can be encoded by strings of polynomial length. In order
to have a uniform treatment, we make the following definition.

\begin{defi}\label{defn:tractable}
Suppose that $\Gamma \in \Seq(\Lambda)$. A set $\Rules$ of one-step
rules is \emph{exponentially tractable} if there is an alphabet
$\Sigma$ and a polynomial $p$ such that every $\flat = (r, \sigma)$
with $r = \Gamma_0 / \Gamma_1 \dots \Gamma_n$
can be encoded as a string of length $\leq p(\size(\Gamma_0
\sigma))$ and the relations
\[ R_1 = \lbrace (\Delta, (\Gamma_0 / \Gamma_1 \dots \Gamma_n,
\sigma)  \mid \Gamma_0 \sigma \subseteq
\Delta \rbrace \]
and 
\[ R_2 = \lbrace ((\Delta, \flat), \Delta') \mid 
\Delta' \mbox{ is $i$-th conclusion of
$\rho(\Delta, \flat)$} \rbrace \]
are decidable in \textsc{Exptime} (modulo this coding) for all $i
\in \Nat$.
\end{defi}
\noindent
Exponential tractability gives an upper bound on the size of the
board of the tableau game and
the complexity of both the parity function and the relation
determining legal moves.
The proof of this result requires the following auxiliary lemmas
thate establish bounds on the closure of the root sequent, and the
size of the sequents in the closure, respectively.

\begin{lem}
\label{lem:clsize}
Suppose $A \in \FoRm(\Lambda)$. Then $|\Cl(A)| \leq \size(A)$.
\end{lem}

\begin{proof}
By induction on the structure of $A$ where the only non-trivial case
is $A = \eta p. B$ for $\eta \in \lbrace \mu, \nu \rbrace$. To
establish the claim, we show that
$D = \lbrace C[p := \eta p. A] \mid C \in
\Cl(B) \rbrace$ is closed. This implies that $\Cl(A) \subseteq D$
and the claim follows from the induction hypothesis.
\end{proof}

\begin{lem}  \label{lemma:sequent-size}
If $\Gamma \in \Seq(\Lambda)$ and $\Delta \in \Seq(\Gamma)$ then
$\size(\Delta) \leq \size(\Gamma)^3$.
\end{lem}

\begin{proof}
The closure of $\Gamma$ has at most $\size(\Gamma)$ many elements,
each of which may be larger than $\size(\Gamma)$ as a result of
substituting $\mu p. A$ for $p$ in $A$ if $\mu p. A \in \Gamma$. The
result follows as this can happen at most $\size(\Gamma)$-many
times.
\end{proof}

We can now formulate, and prove, the annonced encoding of the
tableau game as follows.
\begin{lem} \label{lemma:coding}
Suppose that $\Rules$ is exponentially tractable. Then
every position in the tableau game $\Game_\Gamma  = (B_\exists,
B_\forall, E, \Omega)$ of $\Gamma \in
\Seq(\Lambda)$ can be represented by a string of polynomial length in
$\size(\Gamma)$. Under this coding, the relation $(b, b') \in E$ is decidable in
exponential time.
\end{lem}
\begin{proof}
We know that the state set $A$
of the parity automaton $A$ associated with $\Game_\Gamma$ is
exponential in $\size(\Gamma)$, hence every $a \in A$ can be
represented by a string of polynomial length in $\size(\Gamma)$.

As we are now working with the encoding of the game board
we think of the automaton as operating on encodings of rule blueprints
rather than on the rule blueprints itself. More precisely, we run the automaton not
on trace tiles  $(\Delta,\flat,i)$ but on 
encoded trace tiles $(\mathrm{code}(\Delta),\mathrm{code}(\flat),i)$
where $\mathrm{code}(\Delta)$ is the given encoding of sequents in
$\Seq(\Gamma)$ and $\mathrm{code}(\flat)$ is the encoding of $\flat=(r,\sigma)$
according to Definition~\ref{defn:tractable} 
if $\flat$ encodes a modal rule or $\mathrm{code}(\flat)$ is the principal formula
of the (non-modal) rule represented by $\flat$ otherwise. 

Every element of the set $\Seq(\Gamma)$ can be encoded by a string
of polynomial length in $\size(\Gamma)$ by Lemma
\ref{lemma:sequent-size}. Thus every position $(\Delta, a)$ of
$B_\exists$ can be encoded by a string of polynomial length.

By exponential tractability, every rule blueprint $\flat$ can be
encoded as a string of polynomial length, leading
premise, leading to a polynomial bound on the size of the positions
$(\Delta, \flat, a)$ of $B_\forall$.

To see that $E$ is decidable in exponential time, note that it
follows from exponential tractability that the moves of $\forall$
from $(\Delta, a)$ to $(\Delta, \flat, b)$ are decidable in
\textsc{Exptime} by Definition of tractability. To ensure
\textsc{Exptime} decidablity of a move from $(\Delta, \flat, a)$ to
$(\Delta', a')$ where $\flat$ is a blueprint of a modal rule, note
that the rule represented by $(\Delta, \flat)$ has at most
exponentially many conclusions (measured in the size of $\Delta$),
and as we can check whether $\Delta'$ is the $i$-th conclusion of
$\rho(\Delta, \flat)$ in exponential time, we conclude that $E$ is
decidable in \textsc{Exptime} overall.
\end{proof}

\noindent
%Together with Lemma \ref{lem:convert}, 
We now obtain an 
\textsc{Exptime} upper bound for satisfiability.
\begin{cor} \label{cor:complexity}
Suppose $T$ is a monotone $\Lambda$-structure and $\Rules$ is
exponentially tractable, contraction closed
and  one-step tableau complete for $T$.
Then the problem of deciding whether $\exists$ has a winning
strategy in the tableau game for a clean, guarded
sequent $\Gamma \in \Seq(\Lambda)$ is in \textsc{Exptime}.
As a consequence, the same holds for satisfiability of any
guarded formula $A \in
\FoRm(\Lambda)$.
\end{cor}
\begin{proof} The first assertion follows from Lemma
\ref{lemma:coding}
as the problem of deciding the winner in a parity game 
is exponential only in the size of the parity function of the game 
(Theorem~\ref{fact:paritygames}) which is polynomial 
in the size of $\Gamma$ (Lemma~\ref{lem:parityword}).
%\cite{Klauck:2002:APG}). 
The second statement now follows with the help of Theorem \ref{thm:summary}.
\end{proof}
\begin{exa}
It is easy to see that the rule sets for the modal $\mu$-calculus,
the coalitional $\mu$-calculus and the monotone $\mu$-calculus are
exponentially tractable, as the number of conclusions of each
one-step rule is bounded. To establish exponential tractability for
the rule sets for the graded and probabilistic $\mu$-calculus, we
argue as in \cite{Schroder:2008:PBR} where tractability of the
(dual) proof rules has been established. We encode a rule with
premise
$\sum_{i=1}^n r_i a_i <  k$ as ($r_1, a_1, \dots, r_n, a_n, k)$
and Lemma 6.16 of \emph{op.~cit.} provides 
a polynomial bound on the size of the solutions for the
linear inequalities that combine conclusion and side condition of
both the $(G)$ and $(P)$-rule. Exponential tractability follows,
once we agree on a fixed order on the set of prime implicants.
In all cases, contraction closure is immediate.
\end{exa}

\section{Conclusions}

\noindent
In this paper, we have introduced the coalgebraic $\mu$-calculus
that provides  a generic and
uniform framework for modal fixpoint logics. The calculus takes
three parameters:
\begin{enumerate}[$\bullet$]
\item an endofunctor $T: \Set \to \Set$ that defines the class of
$T$-coalgebras over which the calculus is interpreted
\item a collection $\Lambda$ of modal operators that defines the
syntax of the calculus, and
\item the interpretation of the modal operators over $T$-coalgebras,
which is given by predicate liftings for $T$.
\end{enumerate}
In this general setting, our main results are
soundness and completeness of of the calculus and \textsc{Exptime}
decidability of the satisfiability problem for guarded formulas. Technically,
completeness was achieved by tracking the evolution of fixpoint
formulas in a tableau, and for a closed tableau we require that an
outermost least fixpoint is unfolded along every infinite branch. To
detect these infinite unfoldings of least fixpoints, we use a parity
automaton that we run in parallel with the tableau, so that the
existence of closed tableaux can be characterised by winning
strategies in a parity game that is played on pairs consisting of a
sequent and an automaton state.
Our treatment borrows from by \cite{Niwinski:1996:GMC} and \cite{Venema:2006:AFP}, but
there are some important differences. In contrast to
\cite{Niwinski:1996:GMC}, we use parity games that directly
correspond to tableaux, together with  parity automata to detect bad
traces. Moreover, our model construction super-imposes a coalgebra
structure on the relation induced by a winning strategy for
$\exists$. This model construction is substantially more involved
than that given in \cite{Schroder:2008:PBR}, since we cannot argue
in terms of modal rank in the presence of fixpoints. Compared with
\cite{Venema:2006:AFP} (where no complexity results are presented),
we use standard syntax for modal operators, which allows us to
subsume for instance the graded $\mu$-calculus that cannot be
expressed in terms of the $\nabla$-operator used in \emph{op.~cit.}.
By instantiating the generic approach to specific logics, that is,
by providing instances of the endofunctor $T$, the set $\Lambda$ of
modal operators and the one-step rules $\Rules$, we 
\begin{enumerate}[$\bullet$]
\item reproduce the complexity bound for the modal $\mu$-calculus
\cite{Emerson:1999:CTA}, together with the completeness of a slight
variant of the tableau calculus presented in
\cite{Niwinski:1996:GMC},
\item lead to  a new proof of the known \textsc{Exptime} bound for the graded
$\mu$-calculus \cite{Kupferman:2002:CGM},
\item establish previously unknown \textsc{Exptime} bounds for the 
probabilistic $\mu$-calculus, for coalition logic with fixpoints and
for the monotone $\mu$-calculus.
\end{enumerate}
We note that these bounds are tight for all logics except possibly
the monotone $\mu$-calculus, as the modal $\mu$-calculus can be
encoded into all other logics. 
Given that the coalgebraic framework
is inherently compositional
\cite{Cirstea:2008:CMCS,Cirstea:2004:CAD,Cirstea:2007:MPS,Schroder:2007:CAH},
our results also apply to (coalgebraic) logics that arise by
combining various features, such as strategic games and
quantitative uncertainty.

As mentioned before we would like to stress that we established the EXPTIME bound 
only for the guarded formulas of the above listed logics. Under the frequently used assumption that one can transform an arbitrary formula into an equivalent guarded one in polynomial or even linear time, we could extend our results to the full logics. In particular, note that in \cite{Kupferman:2002:CGM} precisely this assumption has been used for the graded $\mu$-calculus.
For the modal $\mu$-calculus a tableau-based EXPTIME-procedure that works for arbitrary formulas as  input has been presented recently in~\cite{frla:them11}. After careful inspection of our calculus we conjecture that our tableau calculus is also sound and complete for arbitrary formulas and formula sequents. We have to leave the details of the substantially more complicated completeness proof for this general case as future work.

\bibliographystyle{abbrv}
\bibliography{csl2009}

\end{document}